\documentclass[12pt]{article}
\usepackage{amsmath}
\usepackage{amssymb}
\usepackage{graphicx}
\usepackage{color}
\usepackage{algorithm}
\usepackage{hyperref}
\usepackage{cite}
\usepackage{float}
\usepackage{booktabs}
\usepackage{setspace}
\usepackage{authblk}
\usepackage{xcolor}
\usepackage{algorithmic}
\usepackage[margin=1in]{geometry}
\usepackage{times}
\usepackage{titlesec}
\usepackage{subfigure}
\usepackage{multirow}
\usepackage{longtable}
\usepackage{amsfonts}
\usepackage{amsthm}
\usepackage{enumitem}
\usepackage{caption}
\newcommand{\cmt}[1]{\hfill\{#1\}}

\definecolor{stepcolor}{RGB}{255,120,0}

\newcommand{\Step}[1]{\STATE \textbf{\textcolor{stepcolor}{// #1}}}

\definecolor{myorange}{RGB}{255,120,0}
\titleformat{\section}[hang]
{\normalfont\Large\bfseries\color{myorange}}
{\color{myorange}\thesection}
{1em}
{}

\titleformat{\subsection}[hang]
{\normalfont\large\bfseries\color{myorange}}
{\color{myorange}\thesubsection}
{1em}
{}

\titleformat{\subsubsection}[hang]
{\normalfont\normalsize\bfseries\color{myorange}}
{\color{myorange}\thesubsubsection}
{1em}
{}

\newtheorem{theorem}{Theorem}

\usepackage[margin=1in]{geometry}
\usepackage{times}
\title{\color{orange}\textbf{Wheel Dynamic Load Estimation Method Based on Gas Pressure of Hydro-pneumatic Suspension}}

\author[1]{Qijun Liao}
\author[1*]{Jue Yang}
\author[2$\dagger$]{Subhash Rakheja}
\author[1$\dagger$]{Yiting Kang}
\author[2$\dagger$]{Yumeng Yao}
\author[3$\dagger$]{Yuming Yin}

\affil[1]{School of Mechanical Engineering, University of Science and Technology Beijing, Beijing, China}
\affil[2]{Department of Mechanical, Industrial and Aerospace of Engineering, Concordia University, Montreal, Canada}
\affil[3]{College of Mechanical Engineering, Zhejiang University of Technology, Zhejiang, China}
\affil[*]{Contact: yangjue@ustb.edu.cn}

\date{}

\begin{document}
	
	\maketitle
	\renewcommand{\thefootnote}{\fnsymbol{footnote}}
	\footnotetext[1]{Corresponding author}
	\footnotetext[2]{Equal contribution}
	\renewcommand{\thefootnote}{\arabic{footnote}}
	
	\begin{abstract}
		This paper proposes a novel method to estimate the wheel dynamic load based on the gas pressure of a hydro-pneumatic suspension. A nonlinear coupled model between suspension chamber pressure and tire-ground contact force is developed, integrating suspension dynamics with its nonlinear stiffness characteristics. An iterative algorithm is developed to estimate wheel dynamic load using data from only one single pressure sensor, thereby eliminating the reliance on traditional tire models and complex multi-sensor fusion frameworks. This method effectively reduces hardware redundancy and minimizes the propagation of measurement errors. The proposed model is experimentally validated on a dedicated suspension test bench, demonstrating satisfactory agreement between the measured and estimated data. Additionally, co-simulation with TruckSim verifies the accuracy of both the calculated damping force and wheel dynamic load, demonstrating the effectiveness of the model on characterizing the mechanical behavior of the hydro-pneumatic suspension system. The proposed method provides a practical, low-cost, and efficient solution with minimal hardware dependencies.
		
	\end{abstract}
	
	\section{Introduction}
	
	The dynamic performance and load transfer efficiency of vehicle suspension systems have become critical factors in improving overall vehicle handling, safety, and energy efficiency, particularly for distributed electric drive chassis. Wheel dynamic loads not only affect vehicle handling and safety but also play a crucial role in the effectiveness of vehicle stability control systems, anti-lock braking systems (ABS), and autonomous driving algorithms. It is thus essential to accurately estimate dynamic wheel loads to enhance vehicle adaptability and stability.
	
	The dynamic forces generated at the tire–road interface are directly related to the directional stability of the vehicle and the stresses imposed on the road structure. In the context of road loading, heavy freight vehicles are known to induce significant dynamic stresses on pavements, leading to fatigue and premature failure. Furthermore, the roll stability of vehicles with high load capacity and a high center of gravity is directly linked to the transfer of dynamic tire forces, also known as the lateral load transfer ratio (LTR). Real-time monitoring of dynamic tire forces can effectively assist in detecting potential directional instabilities in heavy vehicles, as well as assessing their road-friendliness.
	
	For two-axle vehicles, wheel dynamic loads can be directly calculated through body acceleration. However, for multi-axle chassis, traditional methods are less effective due to increased system complexity. Therefore, there is a growing interest in leveraging suspension system dynamics. Specifically, the output force and extension acceleration of the suspension system for estimating wheel dynamic loads. This estimation process becomes more complex and challenging, especially in multi-axle chassis due to the increased degrees of freedom and interdependencies among suspension components.
	
	Hydro-pneumatic suspensions are widely used in heavy-duty and special-purpose vehicles such as mining trucks, armored vehicles, and high-end passenger cars, owing to their nonlinear stiffness characteristics and energy dissipation performance. The hydro-pneumatic suspension system typically consists of a hydraulic cylinder and a gas chamber, which jointly provide load buffering and attitude adjustment via the coupled action of gas compression and oil damping. However, accurately obtaining the tire-ground pressure (TGP) is a core problem in vehicle stability control. The gas pressure of the hydro-pneumatic suspension system can offer a new approach for estimating tire dynamic loads.
	
	Research on hydro-pneumatic suspension systems has primarily focused on characterizing gas behavior in the suspension chamber and the resulting oil damping force. Many studies describe the relationship between gas pressure and chamber volume using the adiabatic equation $PV^n = C$ assuming the ideal gas behavior. Considering small hole throttling and gap flow phenomena in real-world flow characteristics, many studies developed small hole flow equations or gap flow models to investigate the flow characteristics of gas when the liquid flows through small channels, providing a theoretical basis for calculating suspension system output forces.
	
	For example, Guo et al. \cite{reference1} proposed a new pneumatic suspension system with independently adjustable stiffness and height utilizing double-acting cylinders and accumulators. Yin et al. \cite{reference2} designed a hydro-pneumatic suspension system with variable damping and stiffness, verifying its effectiveness under different excitation conditions through high-speed electromagnetic valves. Zhang et al. \cite{reference3} studied the reliability model of a hydro-pneumatic suspension system on sealing failure and gas leakage issues considering temperature effects. However, none of these studies addressed the direct relationship between gas pressure and tire dynamic loads, especially in multi-axle vehicles. Existing estimation methods often rely on complex suspension models and multi-sensor fusion, which may compromise real-time performance and introduce implementation challenges.
	
	To address these limitations, some researchers have decomposed the suspension system into three subsystems representing gas dynamics, oil throttling, and friction, and then integrated them into comprehensive mechanical models. For instance, Yu et al. \cite{reference4} developed a model incorporated fluid elasticity and high-frequency excitation effects to predict liquid pressure behavior. Cai et al. \cite{reference5} used a 1/4 vehicle model with genetic algorithm to optimize suspension parameters, indirectly verifying the sensitivity of tire dynamic loads to suspension output forces, especially in reducing body acceleration.
	
	To the best of our knowledge, no study to date has established a direct estimation model using hydro-pneumatic suspension gas pressure alone to deduce tire dynamic loads. Alessandro et al. \cite{reference6} proposed an indirect tire state inference method based on wheel speed, and Best et al. \cite{reference7} used tire geometry and pressure data to estimate simple tire dynamic load. Singh et al. \cite{reference8} explored tire dynamic load characteristics based on the extended Pacejka model, while Dallas et al. \cite{reference9} proposed a hierarchical adaptive nonlinear model predictive control method to maximize the utilization of tire dynamic loads.
	
	In the estimation of tire dynamic loads for unmanned vehicles, Liu Zhihao et al. \cite{reference10} developed a finite element model to simulate tire pressure distribution under static ground contact and steady-state rolling conditions, and further established a mathematical model between sinkage, load, and tire pressure to estimate tire dynamic loads. Xu et al. \cite{reference11} used machine learning to develop an intelligent tire system that predicts tire dynamic loads in real-time through accelerometer data and neural networks. Although this method has sufficient accuracy for engineering applications, it requires additional hardware support, increasing implementation complexity. Ma Xiao et al. \cite{reference12} proposed a multi-sensor fusion method that combines accelerometer sensors, tire pressure sensors, and other onboard sensors to obtain wheel dynamic loads in real-time through a Kalman filter algorithm. Although these methods improve tire dynamic load estimation accuracy, they typically rely on high-frequency sensors and complex computation, limiting their practical application in terms of real-time performance and cost.
	
	The previous studies have investigated the influence of tire dynamic loads on vehicle stability and various tire dynamic load estimation techniques based on extended Pacejka models and model predictive control methods. While directly developing the relationship between the gas chamber pressure of the hydro-pneumatic suspension system and tire dynamic loads remains a challenge, traditional methods require complex mathematical models and multi-sensor data fusion. The gas pressure measurement of the hydro-pneumatic suspension can simplify this process. However, the same gas chamber pressure may correspond to different suspension states, such as extension, compression, or varying vibration velocities. Therefore, it is crucial to accurately infer both the output and damping forces of the suspension system and correlate them with changes in wheel dynamic loads. To address this issue, this study proposes a direct method for estimating tire dynamic loads based solely on changes in gas chamber pressure. A detailed mechanical model of the suspension system of a three-axle, six-wheeled vehicle is developed and validated via experimental data. This model is further verified using co-simulation in TruckSim. The simulation results exhibit satisfactory agreement with the experimental data, suggesting the validity and applicability of the proposed method. This research provides a simplified, low-cost method for real-time estimation of tire dynamic loads, with significant potential for application in intelligent driving and unmanned driving systems, enhancing the real-time and accuracy of vehicle stability control strategies. This method avoids the complexity of multi-sensor fusion and reduces system implementation costs while ensuring estimation accuracy, with higher real-time performance and flexibility.

	\section{Suspension Output Force Calculation Based on Gas Pressure}
	
	To avoid symbol ambiguity in subsequent derivations, Table~\ref{tab:nomenclature_merged} uniformly defines the main variables used in this paper.
	
	\begin{longtable}{llll}
		\caption{System Symbol Definitions and Key Parameter Summary} 
		\label{tab:nomenclature_merged} \\
		
		\toprule
		\textbf{Symbol} & \textbf{Definition/Description} & \textbf{Value} & \textbf{Unit} \\
		\midrule
		\endfirsthead
		
		\multicolumn{4}{c}{{\bfseries \tablename\ \thetable{} -- Continued}} \\
		\toprule
		\textbf{Symbol} & \textbf{Definition/Description} & \textbf{Value} & \textbf{Unit} \\
		\midrule
		\endhead
		
		\bottomrule
		\multicolumn{4}{r}{{Continued on next page...}} \\
		\endfoot
		\bottomrule
		\endlastfoot

		\multicolumn{4}{l}{\textit{Part A: General State Variables and Fluid/Gas Properties (Bench \& Simulation)}} \\
		$t, \Delta t$ & Time variable / sampling period & - / $0.002778$ & s \\
		$P_1(t)$ & Main chamber gas pressure (measurable signal) & - & Pa \\
		$P_2(t)$ & Annular chamber hydraulic pressure (state variable) & - & Pa \\
		$h_{\text{total}}, v$ & Suspension total travel / Piston relative velocity & - & m, m/s \\
		$F_{\text{out}}$ & Suspension total output force & - & N \\
		$\rho$ & Hydraulic oil density & $850$ & kg/m$^3$ \\
		$\mu(T_0)$ & Hydraulic oil dynamic viscosity ($30^\circ$C / $50^\circ$C) & $0.065 / 0.032$ & Pa$\cdot$s \\
		$K_{\text{bulk}}$ & Hydraulic oil bulk modulus & $1.7 \times 10^{9}$ & Pa \\
		$\gamma$ & Gas adiabatic index & $1.4$ & - \\
		$n_{\text{eff}}$ & Effective polytropic index & - & - \\
		$P_{\text{atm}}$ & Standard atmospheric pressure & $1.013 \times 10^{5}$ & Pa \\
		\midrule

		\multicolumn{4}{l}{\textit{Part B: Bench Test Prototype Parameters (For Sections 3-4)}} \\
		$D_1$ & Main gas chamber cylinder diameter & $75$ & mm \\
		$D_{\text{piston}}$ & Piston outer diameter & $49$ & mm \\
		$D_2$ & Piston rod diameter & $45$ & mm \\
		$A_1$ & Main piston effective area & $4.418 \times 10^{-3}$ & m$^2$ \\
		$A_2$ & Piston rod cross-sectional area & $1.885 \times 10^{-3}$ & m$^2$ \\
		$A_3$ & Annular chamber effective area ($A_1 - A_2$) & $2.533 \times 10^{-3}$ & m$^2$ \\
		$D_{\text{ch}}$ & Throttle orifice diameter & $6.0$ & mm \\
		$D_{\text{check}}$ & Check valve diameter & $3.0$ & mm \\
		$h_{\text{gap}}$ & Piston-cylinder radial clearance & $0.5$ & mm \\
		$P_0, T_0$ & Initial charge pressure / Initial temperature & $0.8$ / $30, 50$ & MPa, $^\circ$C \\
		\midrule
		
		\multicolumn{4}{l}{\textit{Part C: Heavy Mining Truck Parameters (For Section 6)}} \\
		\multicolumn{4}{l}{\footnotesize \quad \textit{C.1 Mass and Tire Parameters}} \\
		$m_s$ & Sprung mass per wheel & $7500$ & kg \\
		$m_u$ & Unsprung mass per wheel & $800$ & kg \\
		$m_t$ & Tire mass & $500$ & kg \\
		$k_t, c_t$ & Tire vertical stiffness / damping & $2000$ / $4$ & kN/m, kNs/m \\
		$R_{\text{eff}}$ & Tire effective rolling radius & $0.8$ & m \\
		
		\multicolumn{4}{l}{\footnotesize \quad \textit{C.2 Large-Scale Hydro-Pneumatic Suspension}} \\
		$D_{\text{cyl}}$ & Truck suspension cylinder diameter & $250$ & mm \\
		$D_{\text{pis}}$ & Truck piston diameter & $249.5$ & mm \\
		$D_{\text{rod}}$ & Truck piston rod diameter & $210$ & mm \\
		$D_{\text{ch\_truck}}$ & Throttle orifice diameter & $8.0$ & mm \\
		$D_{\text{check\_truck}}$ & Check valve diameter & $8.0$ & mm \\
		$N_{\text{valve}}$ & Valve quantity (Check+Throttle) & $2$ & - \\
		$h_{\text{gap\_truck}}$ & Radial clearance (Truck) & $0.785$ & mm \\
		
		\multicolumn{4}{l}{\footnotesize \quad \textit{C.3 Double-Wishbone Kinematics}} \\
		$L_{\text{upper}}$ & Upper arm effective length & $580$ & mm \\
		$L_{\text{lower}}$ & Lower arm effective length & $650$ & mm \\
		$L_{\text{eff}}$ & Suspension force arm & $480$ & mm \\
		$\beta$ & Suspension axis inclination (Static) & $20$ & deg \\
		$\alpha_0$ & Lower arm installation angle (Static) & $8$ & deg \\
		$\bar{i}_{\text{sus}}$ & Average suspension ratio & $0.755$ & - \\
		
	\end{longtable}
	
	\subsection{Nonlinear Dynamic Modeling of the Suspension System}
	\label{sec:model}
	
	The hydro-pneumatic suspension system consists of a gas spring chamber and a hydraulic damping channel. As shown in Fig.~\ref{fig:suspension_structure}, the system includes a main chamber (gas chamber) and an annular chamber (hydraulic chamber), achieving the coupling effect of gas compression and fluid throttling through piston motion. In the compression stroke, the piston rod moves upward from the equilibrium position, causing the volume of Chamber I to decrease. This compression process increases the gas pressure in Chamber I, while the volume of Chamber II increases and the oil pressure decreases. Consequently, the pressure in Chamber I is higher than in Chamber II, prompting the check valve to open, allowing the oil in Chamber I to flow into Chamber II sequentially through the damping orifice, the check valve, and the piston-cylinder radial clearance. In the extension stroke, the piston rod moves downward from the equilibrium position, the volume of Chamber I increases, and gas pressure decreases; meanwhile, the volume of Chamber II decreases, and oil pressure increases. When the pressure in Chamber II exceeds that in Chamber I, the check valve closes, and oil can only flow from Chamber I to Chamber II through the damping orifice and the piston-cylinder radial clearance.
	
	\begin{figure}[htbp]
		\centering
		\includegraphics[width=0.6\textwidth]{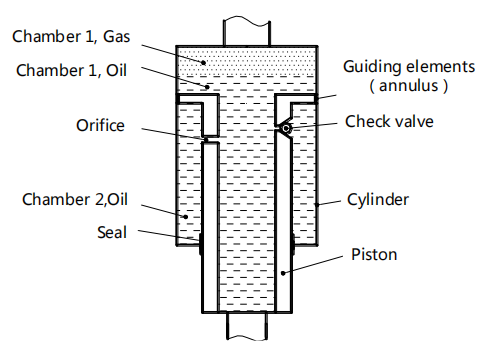}
		\caption{Schematic diagram of hydro-pneumatic suspension system structure}
		\label{fig:suspension_structure}
	\end{figure}
	
	The total output force of the system can be decomposed into three independent mechanical components:
	\begin{equation}
		F_{\text{out}}(t) = F_{\text{gas}}(t) + F_{\text{damp}}(t) + F_{\text{fric}}(t)
		\label{eq:force_decomposition}
	\end{equation}
	where $F_{\text{gas}}$ is the gas pressure (nonlinear stiffness force generated by gas compression), $F_{\text{damp}}$ is the fluid damping force (generated by hydraulic oil flowing through orifices and clearances), and $F_{\text{fric}}$ is the friction force (solid friction between piston and cylinder).
	
	\subsection{Gas Pressure}
	
	\subsubsection{Quasi-Dynamic Polytropic Process}
	
	The state change of gas during the compression-expansion process follows the polytropic process equation:
	\begin{equation}
		P_1 V_{\text{gas}}^{n_{\text{eff}}} = P_0 V_0^{n_{\text{eff}}} = \text{constant}
		\label{eq:polytropic}
	\end{equation}
	where $V_{\text{gas}}$ is the instantaneous gas volume, and $V_0$ is the initial gas volume.
	
	The effective polytropic index $n_{\text{eff}}$ reflects the degree of gas heat exchange. This paper adopts a frequency-based empirical model:
	\begin{equation}
		n_{\text{eff}}(\omega) = 1 + (\gamma - 1) \left[1 - \exp\left(-\frac{\omega}{\omega_c}\right)\right]
		\label{eq:n_eff_freq}
	\end{equation}
	
	At low frequencies ($\omega \ll \omega_c$), the gas has sufficient time for heat exchange, $n_{\text{eff}} \to 1$; at high frequencies ($\omega \gg \omega_c$), heat exchange cannot proceed in time, $n_{\text{eff}} \to \gamma$. Moreover, experimental observations (Section~\ref{sec:temp_stability}) show that the in-cylinder temperature change is extremely small ($\Delta T < 2$°C) during a single excitation process, so it can be considered that:
	\begin{equation}
		n_{\text{eff}}(\omega, T_0) \approx n_{\text{eff}}(\omega) \cdot \left[1 + \alpha_T (T_0 - T_{\text{ref}})\right]
		\label{eq:n_eff_temp}
	\end{equation}
	where $\alpha_T$ is the temperature correction coefficient (calibrated through multi-temperature experiments), and $T_{\text{ref}}$ is the reference temperature (taken as $25$°C).
	
	\subsubsection{Gas Pressure Calculation}
	
	Calculate the current gas volume from Eq.~(\ref{eq:polytropic}):
	\begin{equation}
		V_{\text{gas}}(t) = V_0 \left(\frac{P_0}{P_1(t)}\right)^{1/n_{\text{eff}}(t)}
		\label{eq:V_gas}
	\end{equation}
	
	The gas displacement is:
	\begin{equation}
		h_{\text{gas}}(t) = \frac{V_0 - V_{\text{gas}}(t)}{A_1}
		\label{eq:h_gas}
	\end{equation}
	
	Then the gas pressure force is:
	\begin{equation}
		F_{\text{gas}}(t) = \left[P_1(t) - P_{\text{atm}}\right] A_1 - \left[P_2(t) - P_{\text{atm}}\right] A_2
		\label{eq:F_gas}
	\end{equation}
	
	\subsection{Fluid Damping Force}
	
	\subsubsection{Damping Pressure Difference Composition}
	
	The damping pressure difference comes from the pressure difference between Chamber I and Chamber II. When the piston moves, hydraulic oil flows from the main chamber to the annular chamber (or vice versa), generating pressure drops at the throttle orifice, check valve, and piston clearance:
	\begin{equation}
		\Delta P(t) = P_1(t) - P_2(t) = \Delta P_{\text{visc}} + \Delta P_{\text{inert}} + \Delta P_{\text{orif}} + \Delta P_{\text{gap}}
		\label{eq:pressure_drop}
	\end{equation}
	where $\Delta P_{\text{visc}}$ is viscous pressure drop, $\Delta P_{\text{inert}}$ is inertial pressure drop, $\Delta P_{\text{orif}}$ is orifice throttling pressure drop, and $\Delta P_{\text{gap}}$ is clearance flow pressure drop.
	
	\subsubsection{Viscous Pressure Drop}
	
	The viscous pressure drop is generated by laminar damping, given by the Hagen-Poiseuille law:
	\begin{equation}
		\Delta P_{\text{visc}} = \frac{128 \mu(T_0) L_{\text{ch}} Q}{\pi D_{\text{ch}}^4}
		\label{eq:dp_viscous}
	\end{equation}
	where $L_{\text{ch}}$ is the effective length of the throttle channel.
	
	\subsubsection{Fluid Inertia Effect}
	
	The inertial pressure drop reflects the high-frequency effect of the suspension. Under high-frequency excitation, the acceleration of fluid mass requires additional force, generating inertial pressure drop:
	\begin{equation}
		\Delta P_{\text{inert}} = \rho L_{\text{ch}} \frac{dQ}{dt} / A_{\text{ch}}
		\label{eq:dp_inertial}
	\end{equation}
	This term is negligible in quasi-static processes but becomes one of the dominant factors as excitation frequency increases ($\omega > 3$~Hz), manifesting as the frequency-dependent hysteresis phenomenon observed in bench tests.
	
	\subsubsection{Orifice Throttling Pressure Drop}
	
	The throttling pressure drop when liquid flows through the orifice is given by the Bernoulli equation:
	\begin{equation}
		\Delta P_{\text{orif}} = K_{\text{orif}} \frac{\rho Q^2}{2 A_{\text{eff}}^2} \cdot \text{sign}(Q)
		\label{eq:dp_orifice}
	\end{equation}
	where $K_{\text{orif}}$ is the orifice resistance coefficient, and $A_{\text{eff}}$ is the effective flow area.
	
	The value of $A_{\text{eff}}$ differs between the suspension's compression and extension strokes, which originates from the action of the check valve. In compression stroke ($Q > 0$), the check valve opens, providing additional flow area; in extension stroke ($Q < 0$), the check valve closes, and only the throttle orifice provides flow. The effective areas are respectively:
	\begin{equation}
		A_{\text{eff}} = 
		\begin{cases}
			A_{\text{ch}} + A_{\text{check}}, & Q > 0 \quad \text{(Compression)} \\
			A_{\text{ch}}, & Q \leq 0 \quad \text{(Extension)}
		\end{cases}
		\label{eq:effective_area}
	\end{equation}
	
	\subsubsection{Clearance Flow Pressure Drop}
	
	There is an open wear ring between the piston and cylinder, and the piston-cylinder radial clearance generates a pressure drop:
	\begin{equation}
		\Delta P_{\text{gap}} = \frac{12 \mu(T_0) L_{\text{piston}} Q}{h_{\text{gap}}^3 \pi D_{\text{piston}}}
		\label{eq:dp_gap}
	\end{equation}
	where $L_{\text{piston}}$ is piston length, and $D_{\text{piston}}$ is piston diameter.
	
	\subsubsection{Total Damping Force Calculation}
	
	The annular chamber pressure is obtained by iterative solution:
	\begin{equation}
		P_2(t) = P_1(t) - \Delta P(t)
		\label{eq:P2}
	\end{equation}
	
	The total damping force is:
	\begin{equation}
		F_{\text{damp}}(t) = \Delta P(t) \cdot A_3
		\label{eq:F_damp}
	\end{equation}
	
	\subsection{Friction Force}
	
	The friction between piston and cylinder is described using the Stribeck model:
	\begin{equation}
		F_{\text{fric}}(v) = \left[F_{\text{coulomb}} + (F_{\text{static}} - F_{\text{coulomb}}) \exp\left(-\frac{|v|}{v_{\text{stribeck}}}\right)\right] \tanh(\beta_{\text{fric}} v) + k_v v
		\label{eq:friction}
	\end{equation}
	where $k_v$ is the viscous friction coefficient.
	
	\subsection{Suspension Kinematics Parameter Calculation}
	
	The slight compressibility of the liquid causes its volume change to be:
	\begin{equation}
		\Delta V_{\text{oil}} = \frac{V_{\text{oil},0}}{K_{\text{bulk}}} \Delta P(t)
		\label{eq:oil_compression}
	\end{equation}
	where $V_{\text{oil},0}$ is the initial volume of hydraulic oil.
	
	The total displacement consists of gas displacement and liquid displacement:
	\begin{equation}
		h_{\text{total}}(t) = h_{\text{gas}}(t) + \frac{\Delta V_{\text{gas}}(t) + \Delta V_{\text{oil}}(t)}{A_3}
		\label{eq:h_total}
	\end{equation}
	where $\Delta V_{\text{gas}}(t) = V_0 - V_{\text{gas}}(t)$.
	
	Velocity and acceleration are calculated by numerical differentiation:
	\begin{align}
		v(t) &= -\frac{h_{\text{total}}(t) - h_{\text{total}}(t-\Delta t)}{\Delta t} \label{eq:velocity} \\
		a(t) &= \frac{v(t) - v(t-\Delta t)}{\Delta t} \label{eq:acceleration}
	\end{align}
	where the velocity direction is positive for compression.
	
	\subsection{Complete Algorithm Flow}
	
	\begin{algorithm}[htbp]
		\color{black}
		\caption{Suspension Output Force Calculation Based on Gas Pressure}
		\label{alg:force_calculation}
		\begin{algorithmic}[1]
			\REQUIRE $P_1(t)$: Pressure sequence; $T_0$: Initial temperature
			\ENSURE $F_{\text{out}}(t)$: Suspension output force
			
			\Step{1. Frequency Estimation \& Polytropic Index}
			\STATE $f_{\text{peak}} \gets \text{PeakFrequency}(\text{FFT}(P_1(t)))$
			\STATE $\omega \gets 2\pi f_{\text{peak}}$
			\STATE $n_{\text{eff}} \gets 1 + 0.4(1 - e^{-\omega/12.6})(1 + \alpha_T(T_0 - 25))$
			
			\FOR{$t \in \{t_1, t_2, \ldots, t_N\}$}
			\Step{2. Gas Pressure Calculation}
			\STATE $V_{\text{gas}}(t) \gets V_0 (P_0 / P_1(t))^{1/n_{\text{eff}}}$
			\STATE $h_{\text{gas}}(t) \gets (V_0 - V_{\text{gas}}(t))/A_1$
			
			\Step{3. Velocity \& Flow Rate}
			\STATE $v(t) \gets \text{Diff}(h_{\text{gas}}(t))$
			\STATE $Q(t) \gets A_3 \cdot v(t)$
			
			\Step{4. Fluid Damping Force}
			\IF{$v(t) > 0$ (Compression)}
			\STATE $A_{\text{eff}} \gets A_{\text{ch}}$
			\STATE $\Delta P(t) \gets \frac{128\mu L_{\text{ch}}Q(t)}{\pi D_{\text{ch}}^4} + \rho L_{\text{ch}}\frac{dQ}{dt}/A_{\text{ch}} + K_{\text{orif}}\frac{\rho Q(t)^2}{2A_{\text{ch}}^2} + \frac{12\mu L_{\text{gap}}Q(t)}{2\pi r_p h_{\text{gap}}^3}$
			\ELSE
			\STATE $A_{\text{eff}} \gets A_{\text{ch}} + A_{\text{check}}$
			\STATE $\Delta P(t) \gets K_{\text{orif}}\frac{\rho Q(t)^2}{2A_{\text{eff}}^2} + \frac{12\mu L_{\text{gap}}Q(t)}{2\pi r_p h_{\text{gap}}^3}$
			\ENDIF
			
			\Step{5. Component Force Calculation}
			\STATE $P_2(t) \gets P_1(t) - \Delta P(t)$
			\STATE $F_{\text{gas}}(t) \gets (P_1(t) - P_{\text{atm}})A_1 - (P_2(t) - P_{\text{atm}})A_2$
			\STATE $F_{\text{damp}}(t) \gets \Delta P(t) \cdot A_3$
			\STATE $F_{\text{fric}}(t) \gets [F_{\text{coulomb}} + (F_{\text{static}} - F_{\text{coulomb}})e^{-(v(t)/v_{\text{stribeck}})^2}] \cdot \tanh(\beta_{\text{fric}}v(t))$
			
			\Step{6. Total Output Force}
			\STATE $F_{\text{out}}(t) \gets F_{\text{gas}}(t) + F_{\text{damp}}(t) + F_{\text{fric}}(t)$
			\ENDFOR
			
			\RETURN $F_{\text{out}}(t)$
		\end{algorithmic}
	\end{algorithm}
	
	Figure~\ref{fig:pneumatic suspension working processing} illustrates the schematic workflow of the hydro-pneumatic suspension system, delineating the piston positions at three consecutive discrete time steps: $t_{i}$, $t_{i-1}$, and $t_{i-2}$. These positions correspond to the instantaneous compression or extension states of the suspension, serving as the fundamental kinematic inputs for the subsequent force estimation. Building upon this physical process, Algorithm~\ref{alg:force_calculation} presents the comprehensive computational framework for estimating the suspension output force based solely on gas pressure signals. This algorithm integrates a real-time frequency estimation module to dynamically adjust the polytropic index, ensuring accurate force calculation across varying excitation frequencies.
	
	\begin{figure}[htbp]
		\centering
		\includegraphics[width=0.75\textwidth]{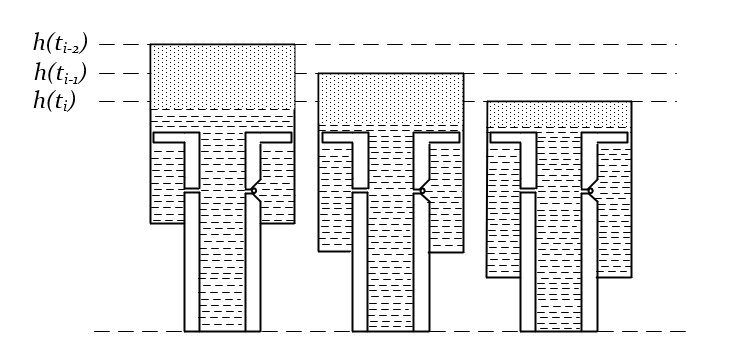}
		\caption{Working process diagram of hydro-pneumatic suspension}
		\label{fig:pneumatic suspension working processing}
	\end{figure}
	
	\section{Bench Test Validation}
	\label{sec:bench_test}
	
	\subsection{Experimental System Configuration and Test Conditions}
	
	Experiments were conducted on an MTS hydraulic servo test bench. The test system includes an MTS hydraulic excitation system (maximum excitation force $\pm 50$~kN, frequency range $0.1$-$20$~Hz), force sensor (accuracy $\pm 0.5\%$~FS), pressure sensor (accuracy $\pm 0.1\%$~FS), displacement sensor (LVDT type, accuracy $\pm 0.01$~mm), and temperature sensor (K-type thermocouple, accuracy $\pm 0.5$°C).
	
	Figure~\ref{fig:suspension-test-rig} shows the layout of the test bench. As shown, the top of the cylinder is fixed to a crossbeam, and the vibrator applies excitation to the bottom of the piston. A load sensor is installed between the crossbeam and the cylinder to measure the suspension output force. Pressure sensor 1 is installed at the bottom of the piston to measure the oil pressure in Chamber I (gas chamber), which also represents the gas pressure. Pressure sensor 2 is installed in the annular chamber to measure the oil pressure in Chamber II. In addition, the vibrator system is equipped with displacement and velocity sensors to achieve comprehensive dynamic data acquisition. The experiment collected at least 20 complete cycles to eliminate initial transient effects, and the data used for analysis were continuous 2-3 cycles in the steady-state phase. All sensor signals used a Butterworth low-pass filter with a 50Hz cutoff frequency to eliminate high-frequency measurement noise, but retained the main dynamic characteristics of the system. The slight drift of trajectories after multiple cycles mainly stems from: (i) minute cumulative changes in hydraulic oil temperature; (ii) slight non-repeatability of the gas polytropic process. This drift amplitude is extremely small and does not affect the validity of model verification.
	
	\begin{figure}[htbp]
		\centering
		\includegraphics[width=0.9\textwidth]{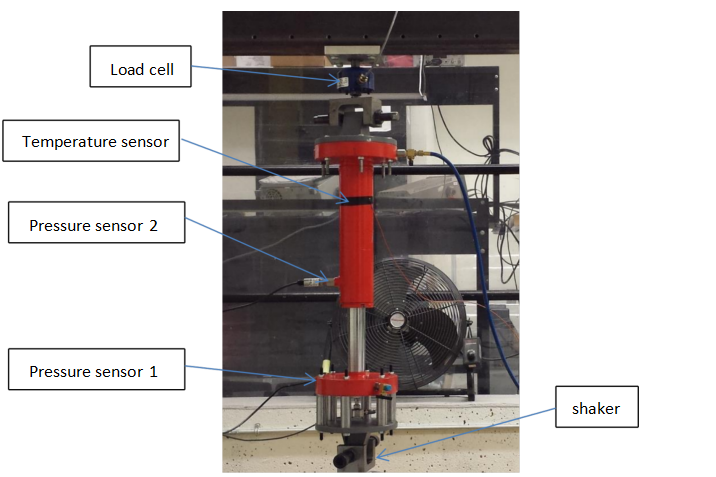}
		\caption{Suspension on test rig}
		\label{fig:suspension-test-rig}
	\end{figure}
	
	Table~\ref{tab:nomenclature_merged} summarizes the parameters of the suspension test piece. To systematically evaluate the fidelity of the model under different thermodynamic states and dynamic frequency domains, a full factorial test matrix containing two temperature levels ($30^\circ$C and $50^\circ$C) was designed. Under each temperature condition, four sinusoidal displacement excitations were applied respectively: low frequency band ($3$~Hz and $5$~Hz) maintained a constant amplitude of $7.50$~mm; high frequency band amplitude decreased with increasing frequency to adapt to the dynamic performance limit of the bench, with amplitudes for $7$~Hz and $8$~Hz set to $5.36$~mm and $4.69$~mm respectively. The above combinations formed 8 sets of steady-state test conditions, effectively covering the typical working range from quasi-static to high-frequency inertia-dominated.
	
	To verify the accuracy of the calculation model under different temperatures and frequencies, systematic test conditions were designed. The rationality of using sinusoidal displacement excitation lies in: (1) Sinusoidal signals are standard inputs for linear and nonlinear system identification, allowing systematic investigation of frequency response characteristics; (2) Although actual road excitation is a random signal, it can be decomposed into the superposition of different frequency components, and the frequency domain characteristics of the suspension system can be extrapolated through harmonic response. Bench tests were performed under two temperatures ($30$°C and $50$°C) and four frequencies ($3$, $5$, $7$, $8$~Hz), obtaining $2 \times 4 = 8$ sets of complete data. For each set of data, three types of characteristic curves can be plotted: output force-gas pressure, velocity-gas pressure, displacement-gas pressure.
	
	\subsection{Temperature Monitoring and Stability}
	\label{sec:temp_stability}
	
	To verify the assumption of "temperature basically unchanged during operation," the in-cylinder temperature was synchronously collected and recorded in each experiment. Figure~\ref{fig:temp_monitoring} shows the temperature changes under all working conditions. It can be seen from the figure that under all experimental conditions, the in-cylinder temperature fluctuation was $< 1.6$°C (peak-to-peak), and the mean drift was $< 1$°C, which can be regarded as almost constant temperature during a long enough period. This indicates that heat exchange is insufficient to significantly change the system temperature, so it is reasonable to treat the initial temperature $T_0$ as a constant. The temperature effect is mainly reflected in the comparison of output characteristics under different initial temperature settings.
	
	\begin{figure}[htbp]
		\centering
		\subfigure[Cylinder temperature variation at initial 30°C]{
			\includegraphics[width=0.48\textwidth]{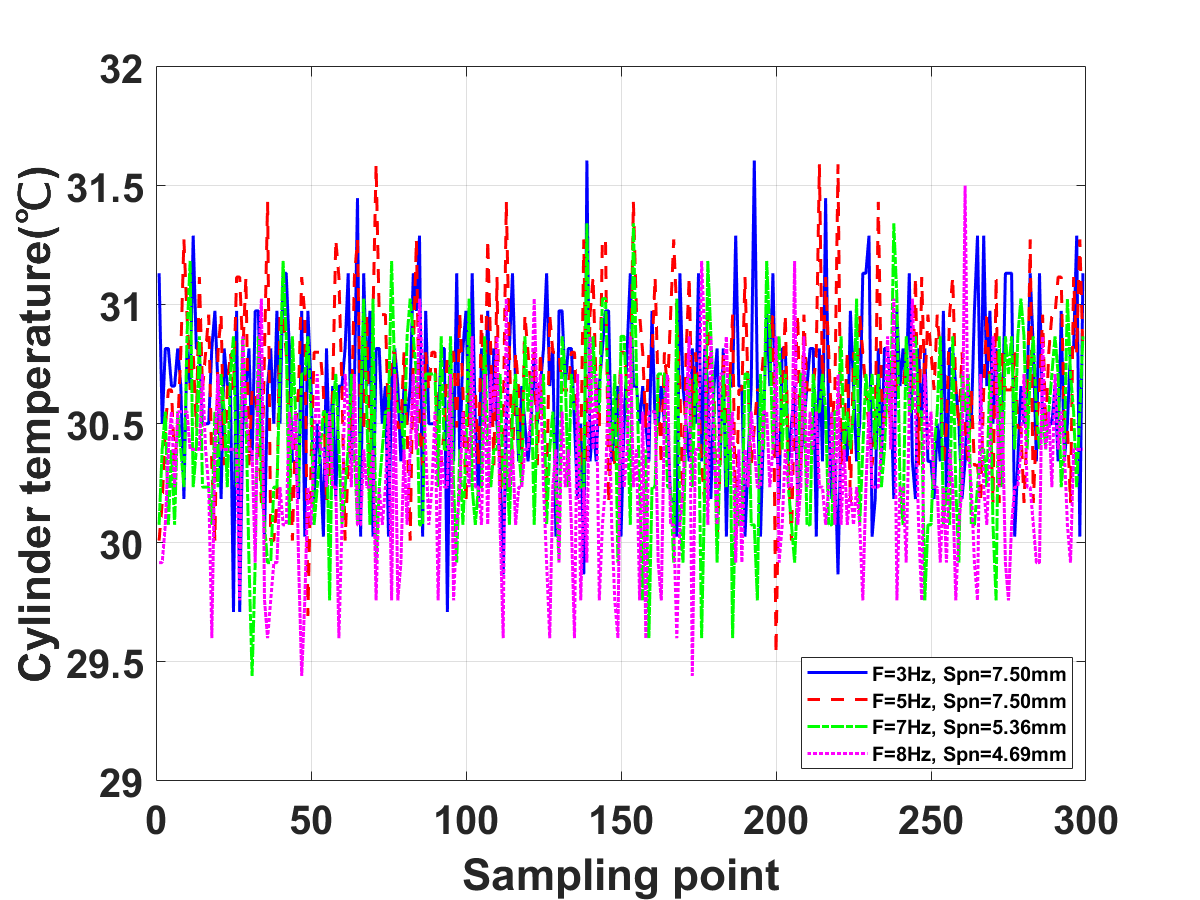}
			\label{fig:Temperature_Fluctuate_30degC}
		}
		\hfill
		\subfigure[Cylinder temperature variation at initial 50°C]{
			\includegraphics[width=0.48\textwidth]{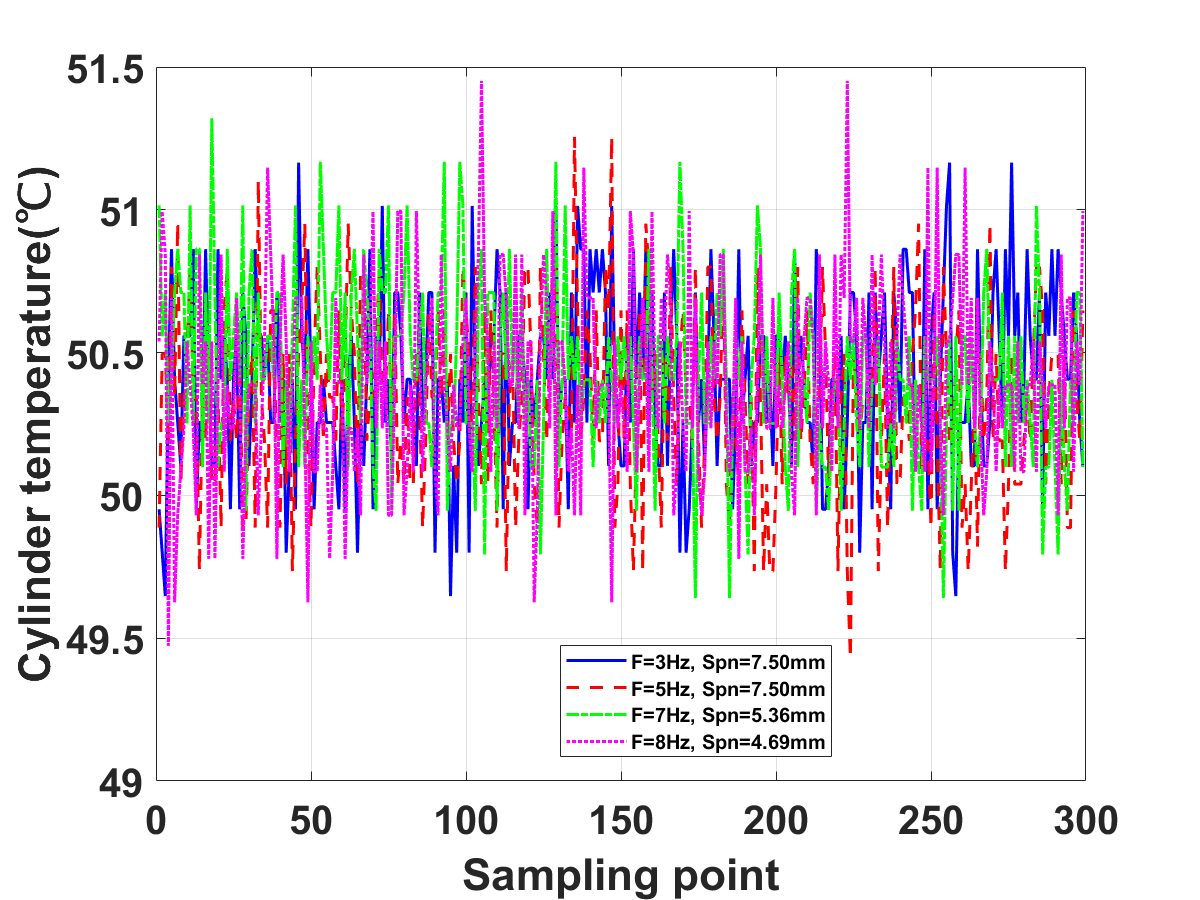}
			\label{fig:Temperature_Fluctuate_50degC}
		}
		\caption{Cylinder temperature variation under all experimental conditions}
		\label{fig:temp_monitoring}
	\end{figure}
	
	\subsection{Model Output Results Validation}
	
	\subsubsection{Output Force-Gas Pressure}
	
	Figure~\ref{fig:force_comparison} shows the comparison of suspension output force-gas pressure relationships at four frequencies (3, 5, 7, 8~Hz) under 30°C and 50°C. In the figure, blue "+" are experimental measurement values, and red "--" are algorithm estimation values.
	
	\begin{figure}[htbp]
		\centering
		\subfigure[30°C, 3~Hz]{
			\includegraphics[width=0.23\textwidth]{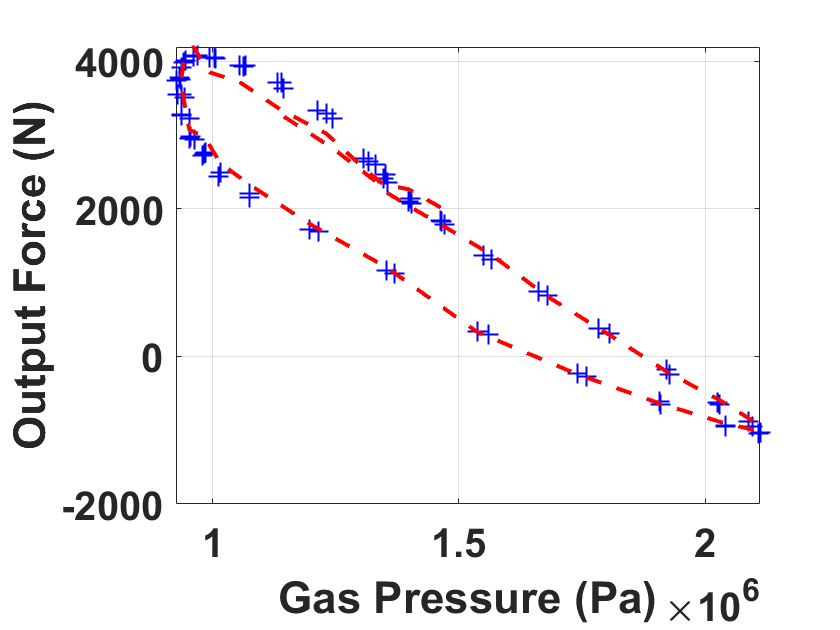}
		}
		\hfill
		\subfigure[30°C, 5~Hz]{
			\includegraphics[width=0.23\textwidth]{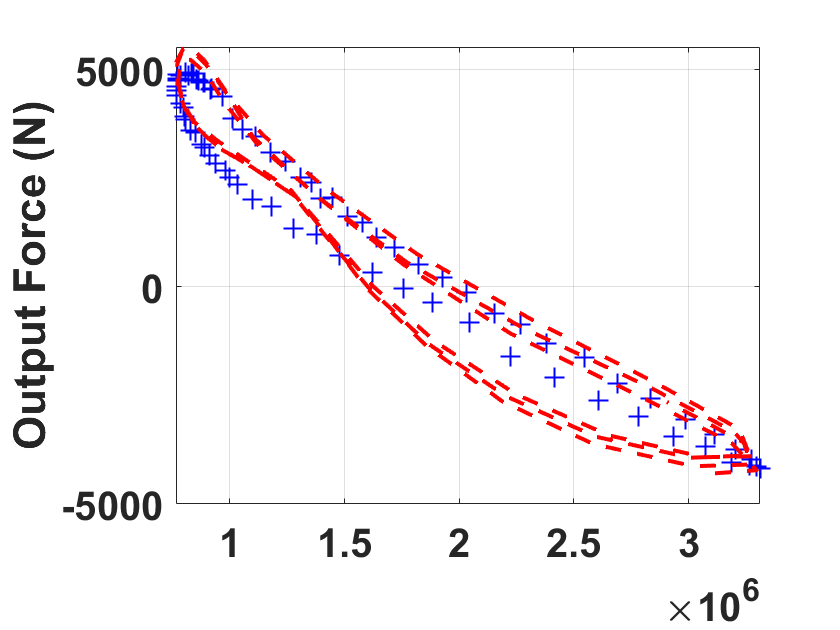}
		}
		\hfill
		\subfigure[30°C, 7~Hz]{
			\includegraphics[width=0.23\textwidth]{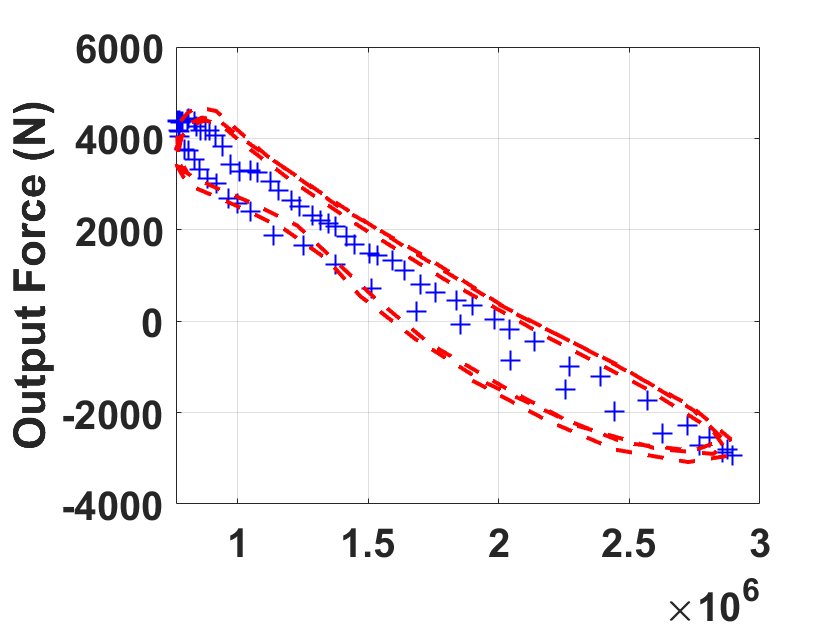}
		}
		\hfill
		\subfigure[30°C, 8~Hz]{
			\includegraphics[width=0.23\textwidth]{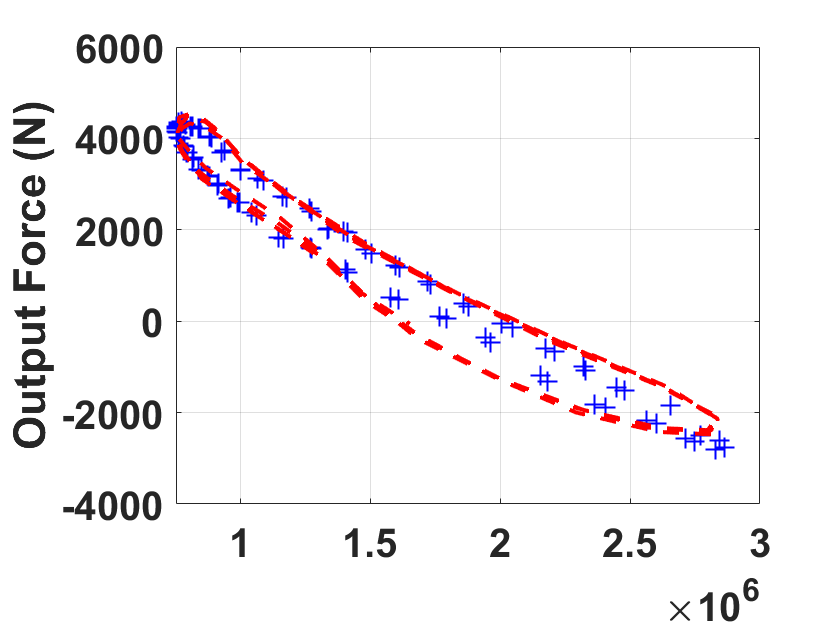}
		}
		\vfill
		\subfigure[50°C, 3~Hz]{
			\includegraphics[width=0.23\textwidth]{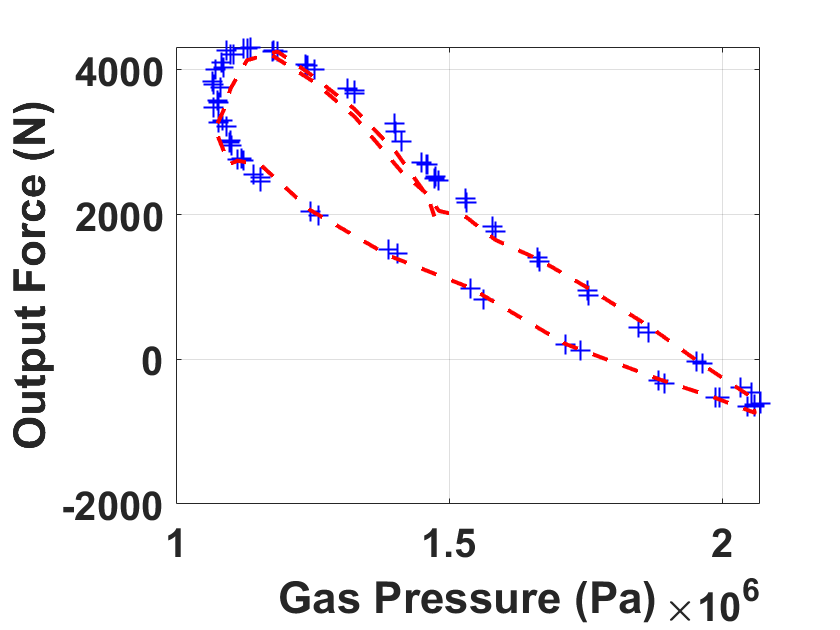}
		}
		\hfill
		\subfigure[50°C, 5~Hz]{
			\includegraphics[width=0.23\textwidth]{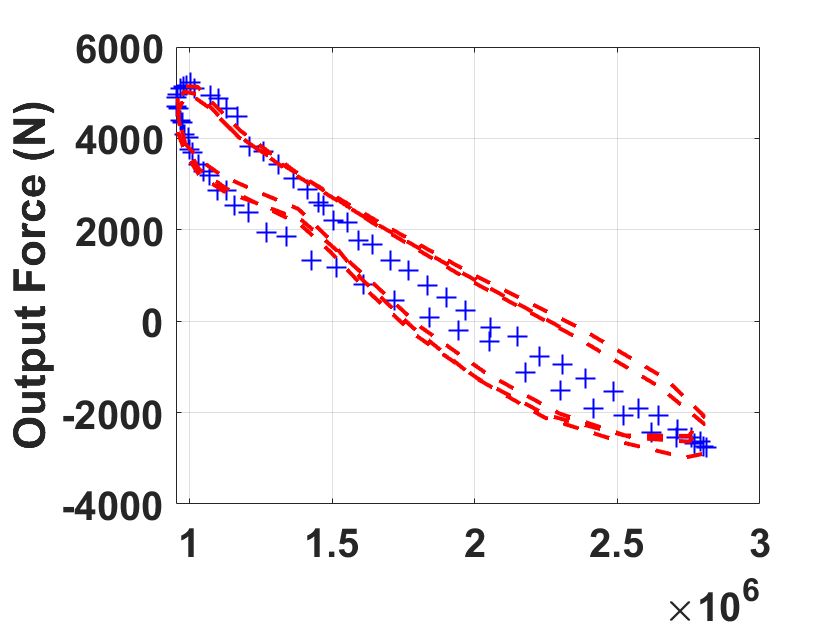}
		}
		\hfill
		\subfigure[50°C, 7~Hz]{
			\includegraphics[width=0.23\textwidth]{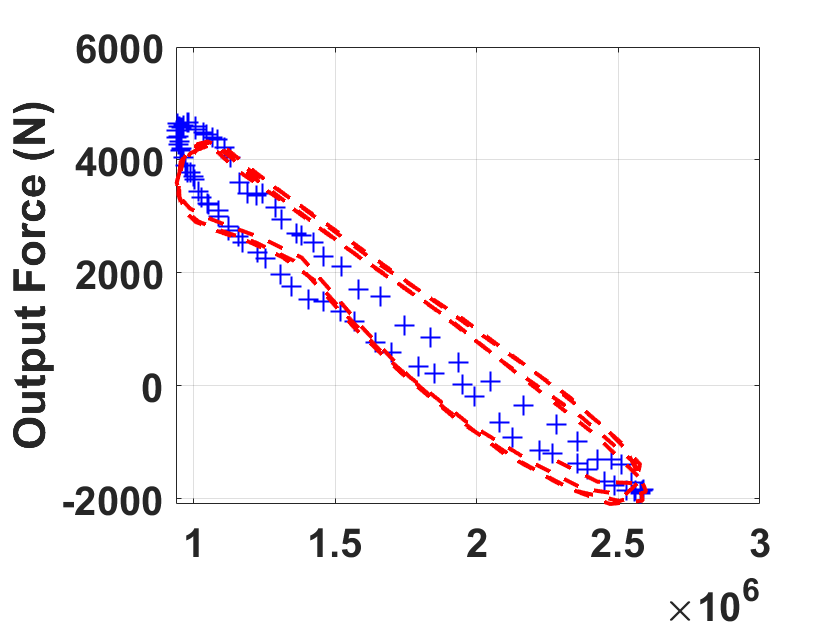}
		}
		\hfill
		\subfigure[50°C, 8~Hz]{
			\includegraphics[width=0.23\textwidth]{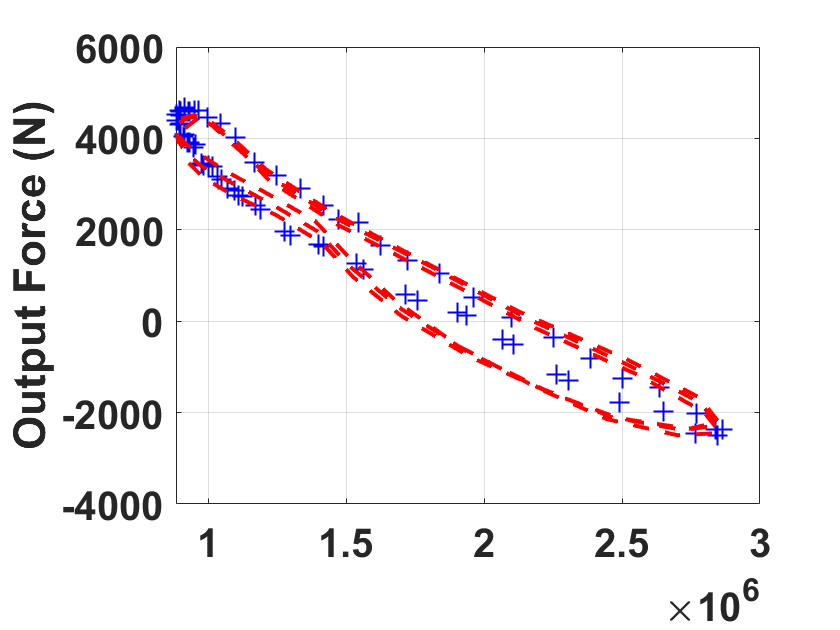}
		}
		\caption{Suspension output force-pressure relationship at 30°C and 50°C initial temperatures}
		\label{fig:force_comparison}
	\end{figure}
	
	From the image observation, it can be seen that the red estimated curve and the blue experimental data show good consistency in the compression stroke (lower branch) and extension stroke (upper branch) under all working conditions, and conform to the introduction expectation of inertial pressure drop (Eq.~\ref{eq:dp_inertial}), both showing hysteresis, presenting the following characteristics:
	
	(i) Low frequency condition (3~Hz): Figures~\ref{fig:force_comparison}(a) and (e) show that the blue points and red lines almost completely coincide, and the deviation is mainly concentrated in the extension stroke near the pressure extremum region ($1.0 \times 10^6 < P < 1.3 \times 10^6$~Pa), with a maximum deviation of about $\pm 200$~N, and the hysteresis effect begins to appear.
	
	(ii) Medium frequency condition (5~Hz): Figures~\ref{fig:force_comparison}(b) and (f) show that the compression stroke and extension stroke curves are slightly close. Comparing 30°C and 50°C, the enclosed area of the latter is smaller (reduced by about 10\%), reflecting the weakening of damping force caused by viscosity reduction ($\mu(50^\circ\text{C}) / \mu(30^\circ\text{C}) \approx 0.49$, Table~1). The estimated curve fits best in the medium and low pressure region ($1.0$-$1.7 \times 10^6$~Pa), with deviation $< 150$~N.
	
	(iii) High frequency condition (7-8~Hz): Figures~\ref{fig:force_comparison}(c)(d) and (g)(h) show that the closed curve formed by the blue experimental points is tighter and the enclosed area is smaller, and the red estimated curve captures this change very well. Although the dispersion of data points at 50°C is slightly larger than that at 30°C (standard deviation increases by about 5\%), the overall fitting degree of the curve is still high.
	
	(iv) Temperature effect: Comparing the upper and lower rows of Figure~\ref{fig:force_comparison}, under the same frequency and pressure, the output force at 50°C is almost unchanged, with a very small increase at high frequency. Under 8~Hz condition, the peak force at 30°C is about 4000~N (Figure~\ref{fig:force_comparison}(d), $P \approx 1.0 \times 10^6$~Pa), and it increases to about 4100~N at 50°C (Figure~\ref{fig:force_comparison}(h), same pressure), an increase of about 2.5\%, which is consistent with the theoretical expectation of reduced viscous damping but enhanced inertial pressure drop.
	
	\subsubsection{Velocity-Gas Pressure and Position-Gas Pressure}
	
	Figures~\ref{fig:phase_plots_30degC} and \ref{fig:phase_plots_50degC} respectively show the velocity-gas pressure and position-gas pressure phase diagrams of four frequency (3, 5, 7, 8~Hz) conditions under 30°C and 50°C. In the figure, blue "+" are experimental measurement values, and red "--" are algorithm estimation values.
	
	\begin{figure}[htbp]
		\centering
		\subfigure[Velocity-Pressure (3~Hz)]{
			\includegraphics[width=0.23\textwidth]{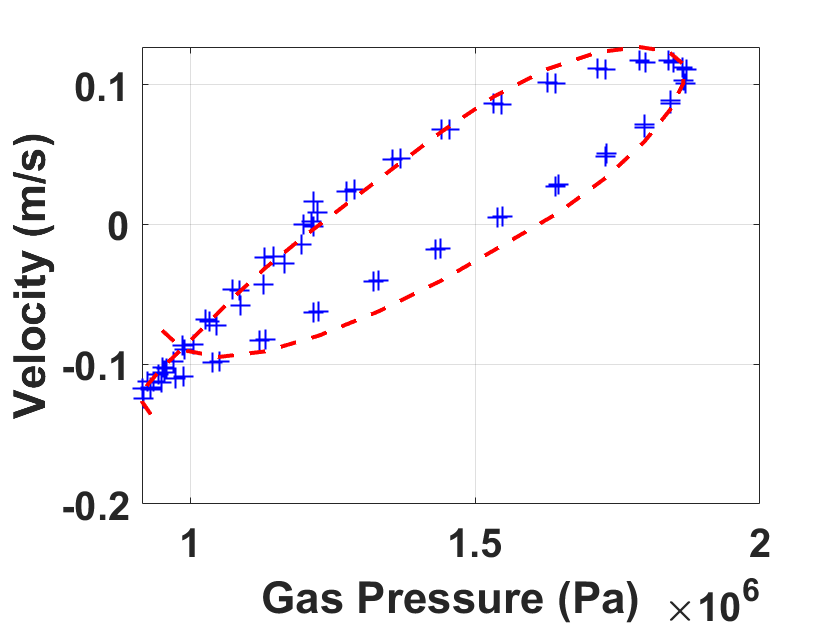}
		}
		\hfill
		\subfigure[Velocity-Pressure (5~Hz)]{
			\includegraphics[width=0.23\textwidth]{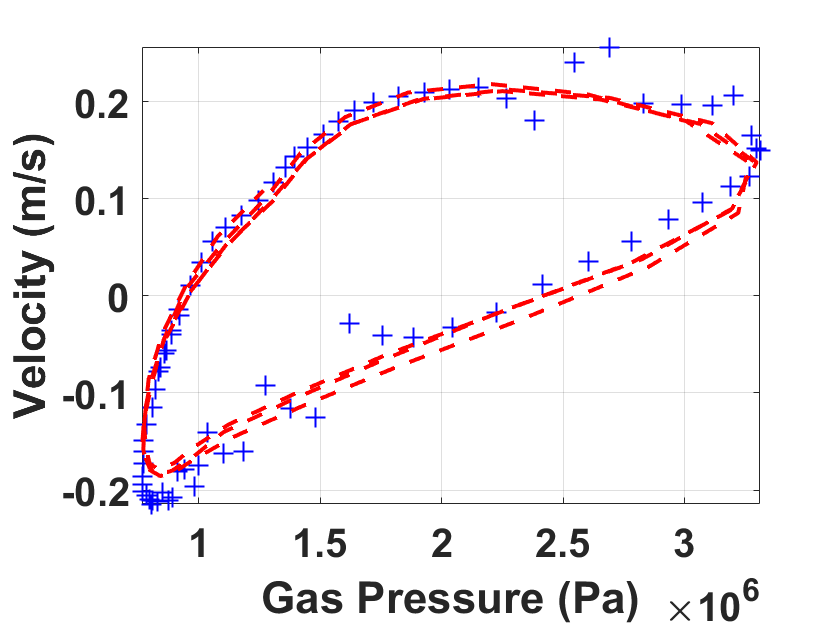}
		}
		\hfill
		\subfigure[Velocity-Pressure (7~Hz)]{
			\includegraphics[width=0.23\textwidth]{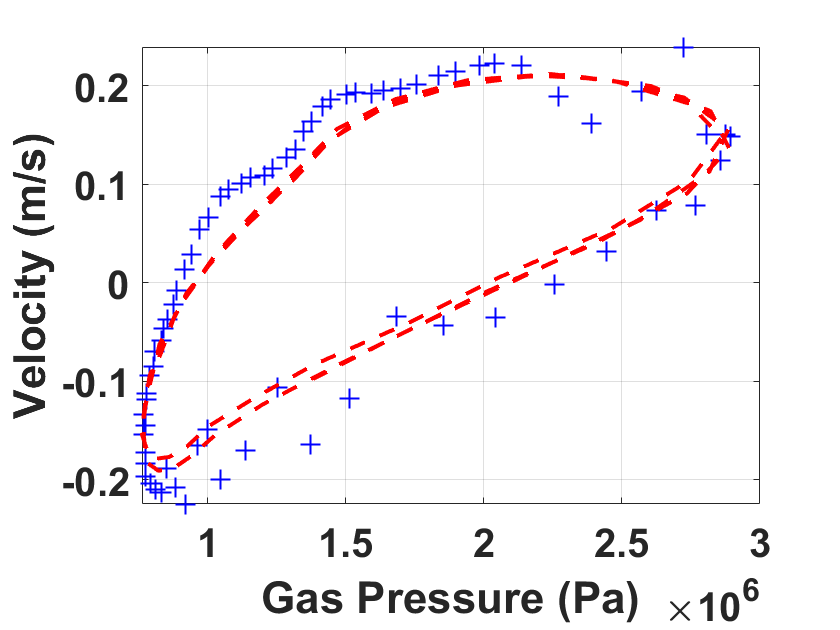}
		}
		\hfill
		\subfigure[Velocity-Pressure (8~Hz)]{
			\includegraphics[width=0.23\textwidth]{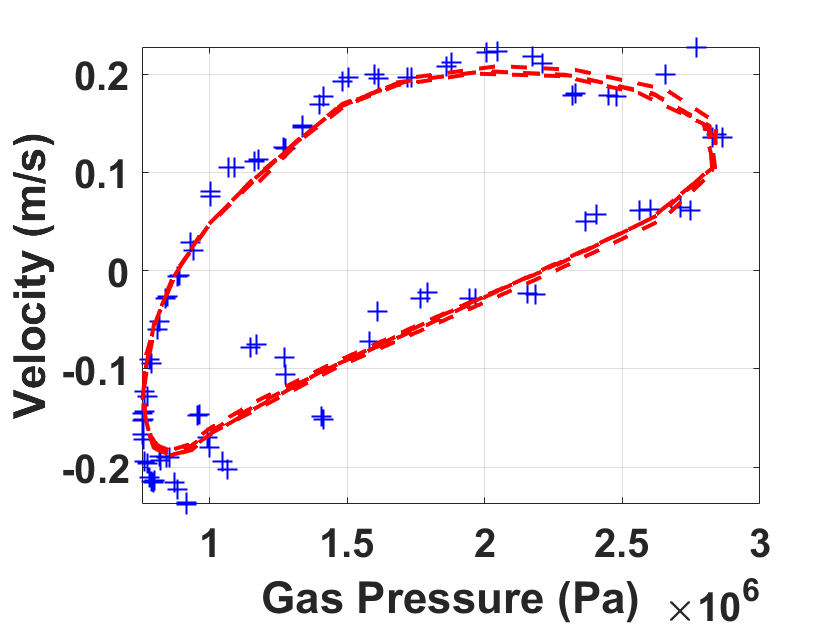}
		}
		\vfill
		\subfigure[Position-Pressure (3~Hz)]{
			\includegraphics[width=0.23\textwidth]{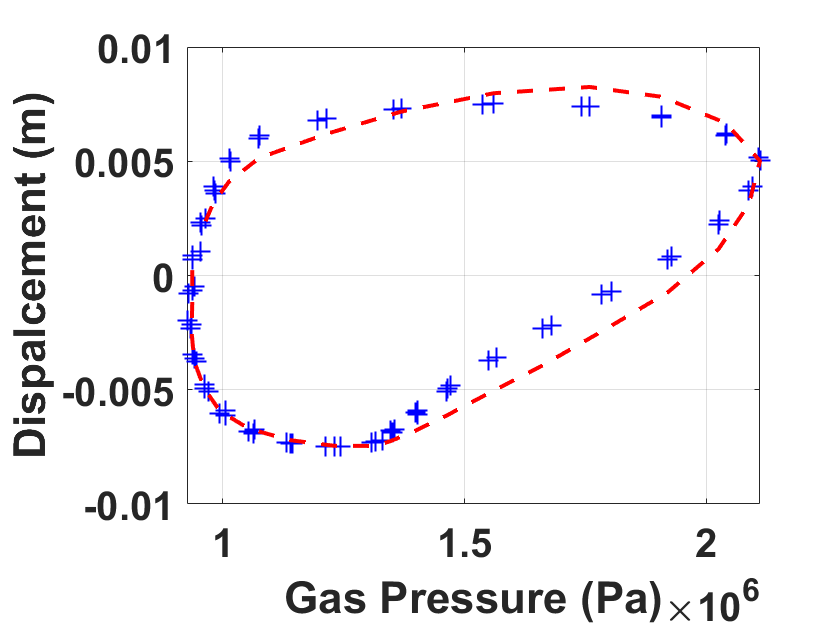}
		}
		\hfill
		\subfigure[Position-Pressure (5~Hz)]{
			\includegraphics[width=0.23\textwidth]{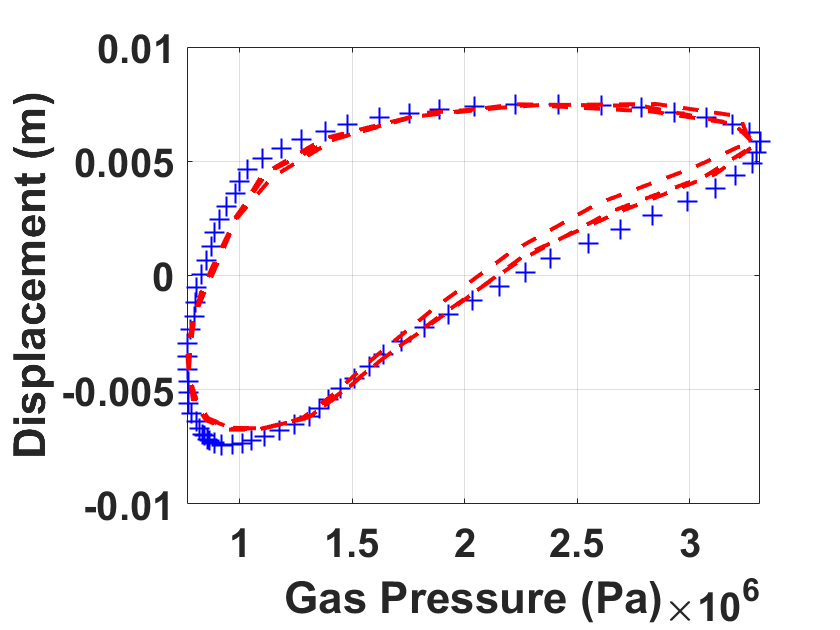}
		}
		\hfill
		\subfigure[Position-Pressure (7~Hz)]{
			\includegraphics[width=0.23\textwidth]{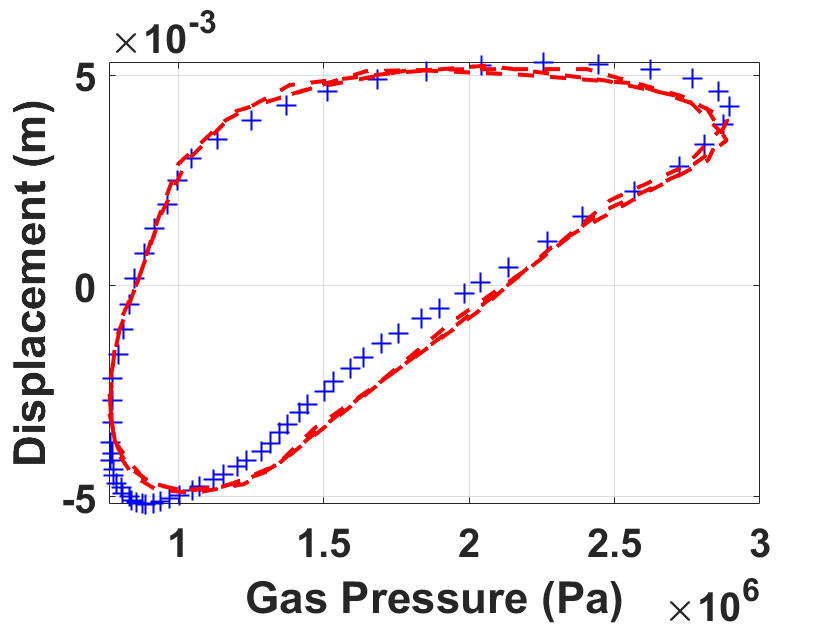}
		}
		\hfill
		\subfigure[Position-Pressure (8~Hz)]{
			\includegraphics[width=0.23\textwidth]{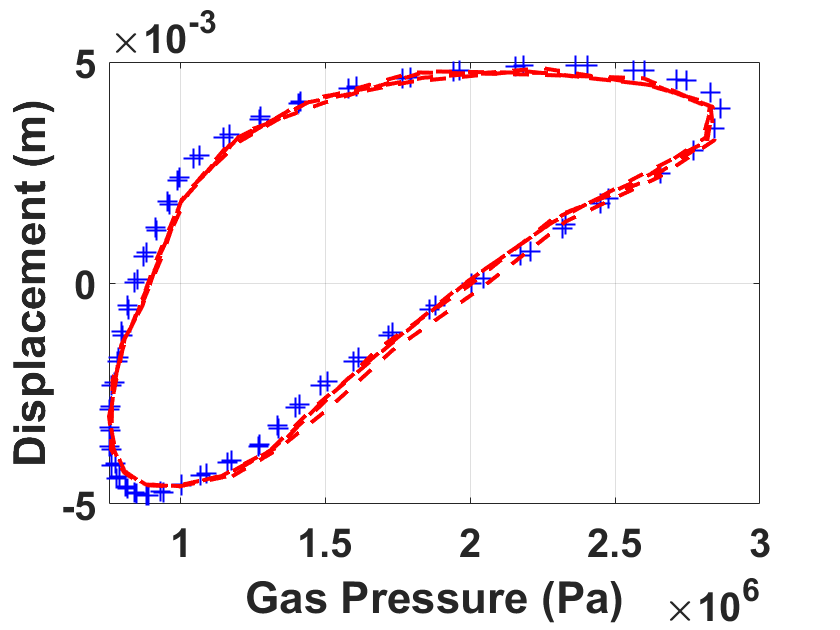}
		}
		\caption{Velocity-pressure and position-pressure relationships at 30°C initial temperature}
		\label{fig:phase_plots_30degC}
	\end{figure}
	
	\begin{figure}[htbp]
		\centering
		\subfigure[Velocity-Pressure (3~Hz)]{
			\includegraphics[width=0.23\textwidth]{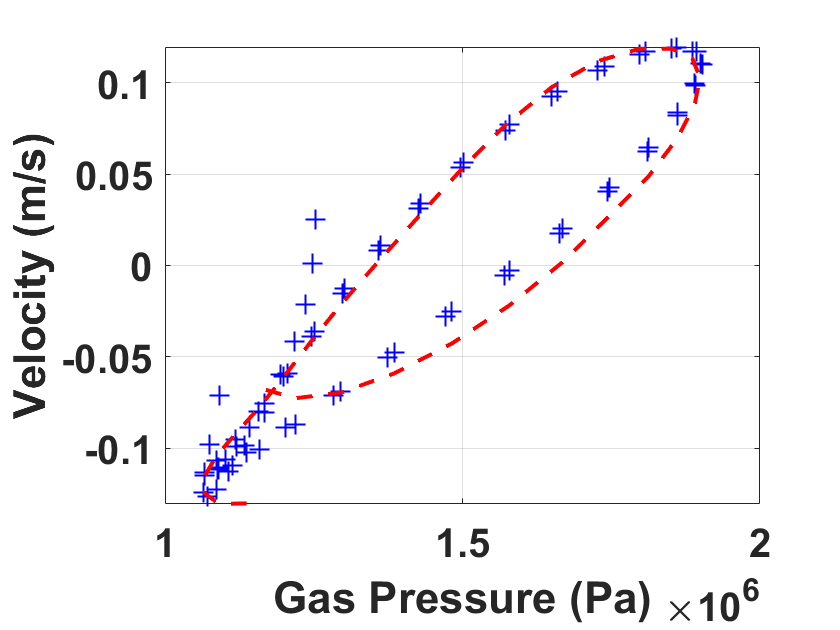}
		}
		\hfill
		\subfigure[Velocity-Pressure (5~Hz)]{
			\includegraphics[width=0.23\textwidth]{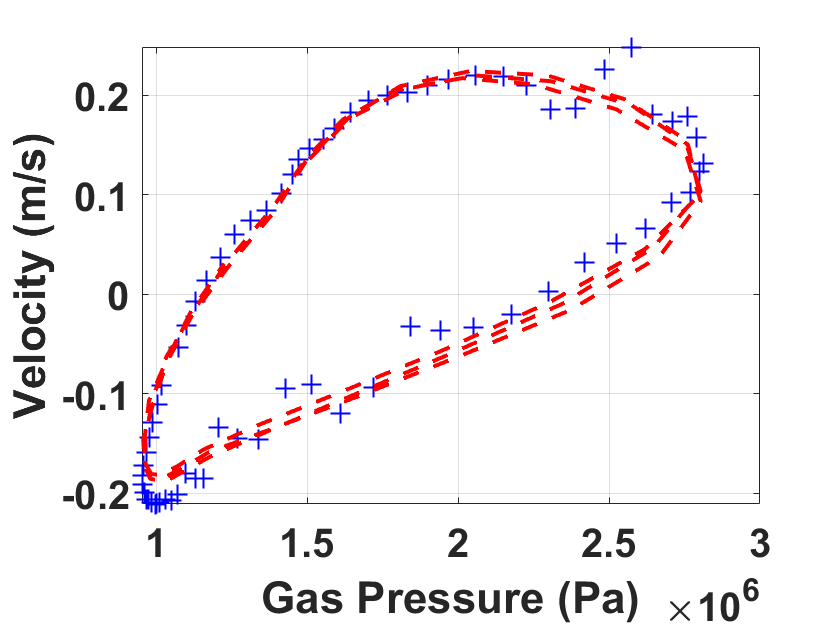}
		}
		\hfill
		\subfigure[Velocity-Pressure (7~Hz)]{
			\includegraphics[width=0.23\textwidth]{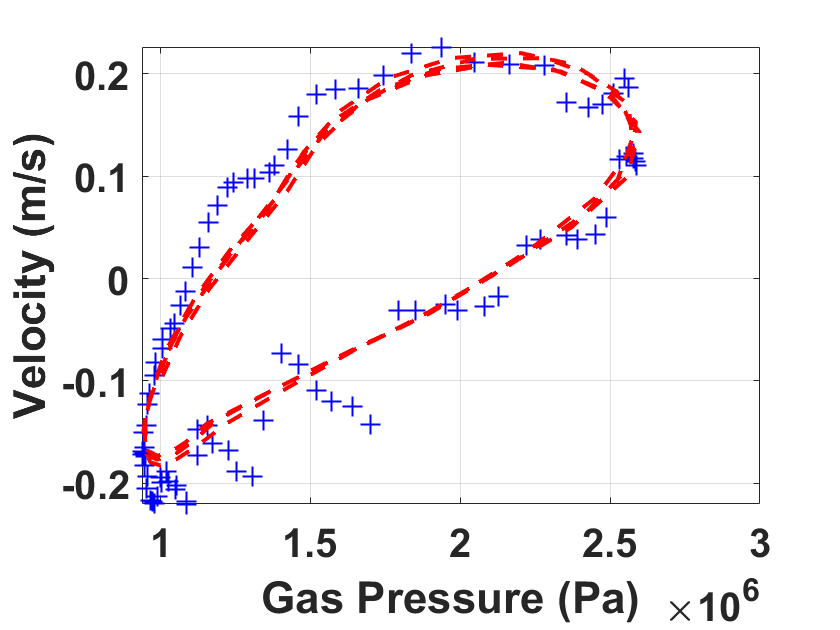}
		}
		\hfill
		\subfigure[Velocity-Pressure (8~Hz)]{
			\includegraphics[width=0.23\textwidth]{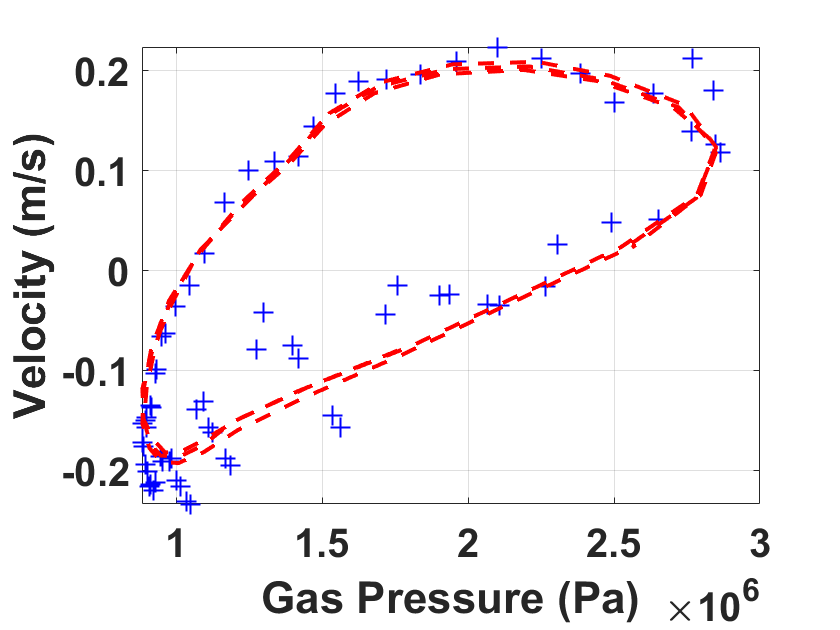}
		}
		\vfill
		\subfigure[Position-Pressure (3~Hz)]{
			\includegraphics[width=0.23\textwidth]{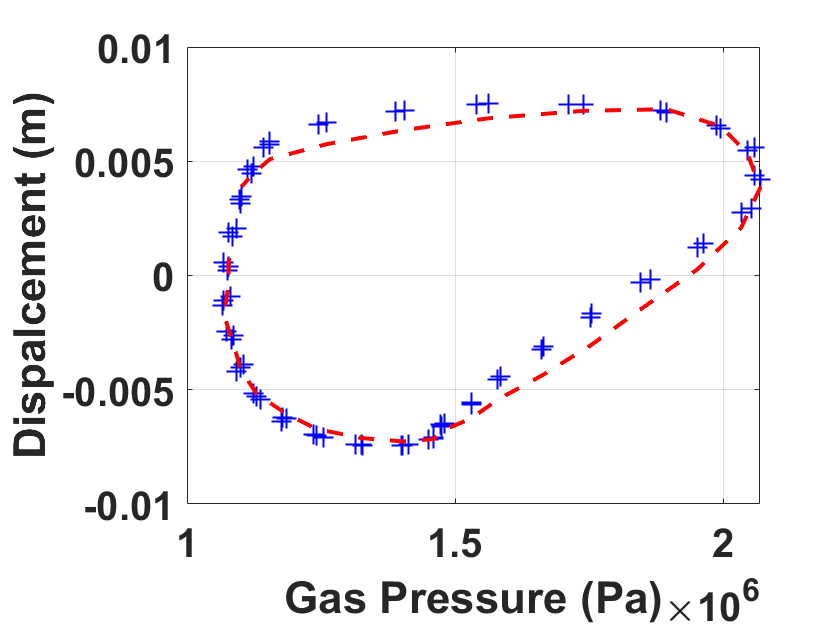}
		}
		\hfill
		\subfigure[Position-Pressure (5~Hz)]{
			\includegraphics[width=0.23\textwidth]{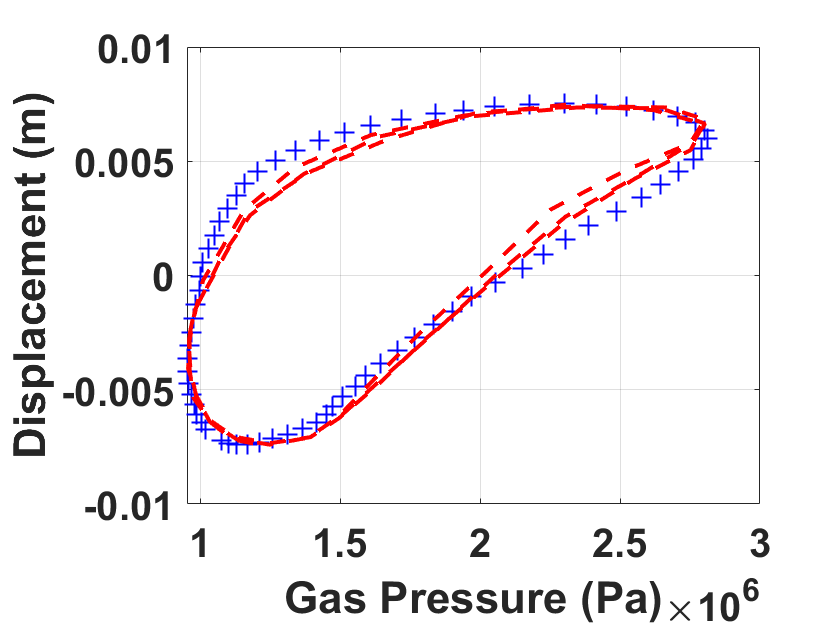}
		}
		\hfill
		\subfigure[Position-Pressure (7~Hz)]{
			\includegraphics[width=0.23\textwidth]{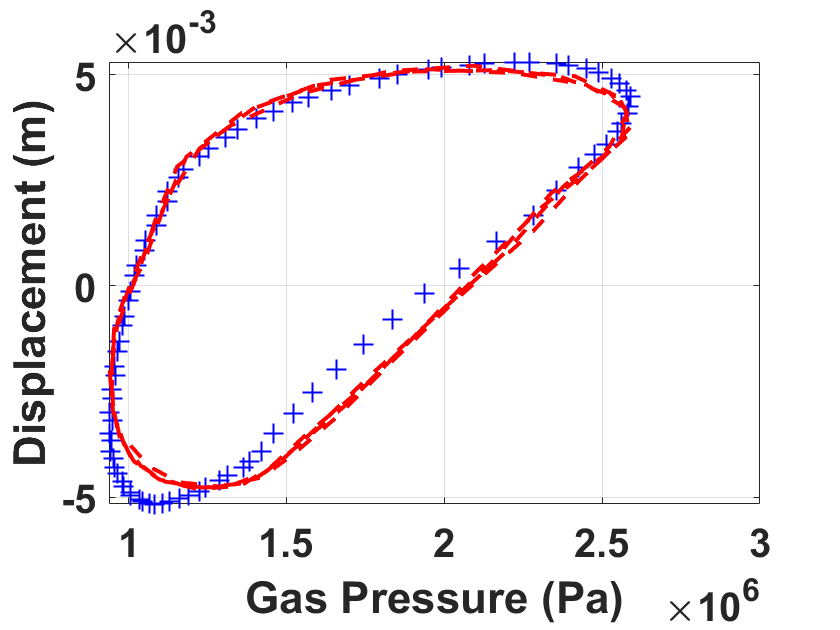}
		}
		\hfill
		\subfigure[Position-Pressure (8~Hz)]{
			\includegraphics[width=0.23\textwidth]{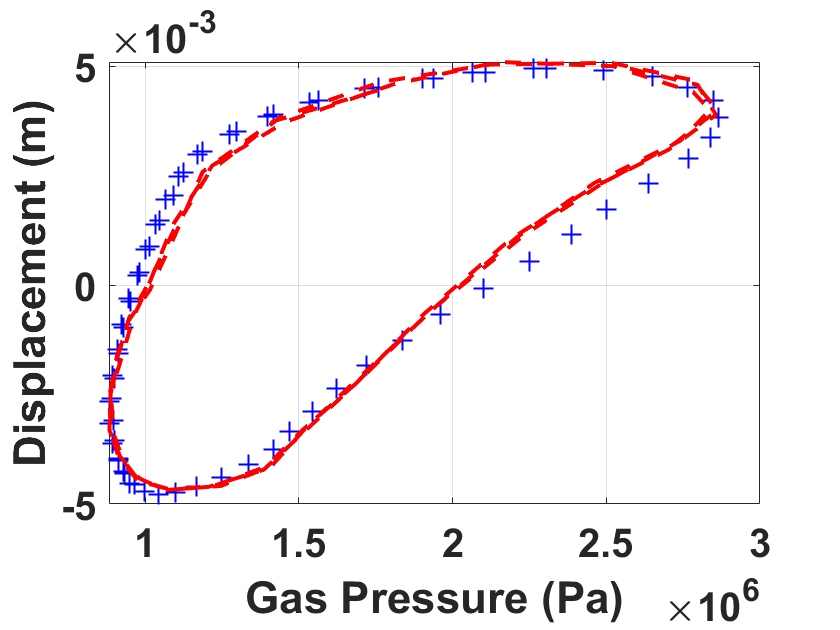}
		}
		\caption{Velocity-pressure and position-pressure relationships at 50°C initial temperature}
		\label{fig:phase_plots_50degC}
	\end{figure}
	
	Velocity-Gas Pressure Diagram Analysis:
	
	Under 3~Hz low frequency condition, Figures~\ref{fig:phase_plots_30degC}(a) and \ref{fig:phase_plots_50degC}(a) show that the velocity-pressure trajectory presents a relatively symmetrical narrow ellipse in both compression and extension strokes. The red estimated curve coincides highly with the blue experimental points, and the deviation is mainly concentrated in the compression stroke (about $-0.01$~m/s). The elliptical shape reflects that although the system has hysteresis response, it is still relatively close to a quasi-static process.
	
	As frequency increases to 8~Hz (Figures~\ref{fig:phase_plots_30degC}(d) and \ref{fig:phase_plots_50degC}(d)), the trajectory obviously widens and presents an asymmetrical "half-moon shape". This morphological change directly reflects the significant enhancement of inertial pressure drop $\Delta P_{\text{inert}} = \rho L_{\text{ch}} (dQ/dt) / A_{\text{ch}}$—fluid acceleration at high frequency requires additional pressure difference to overcome inertial resistance, leading to a larger velocity range corresponding to the same pressure. Comparing Figure~\ref{fig:phase_plots_30degC}(d) and Figure~\ref{fig:phase_plots_50degC}(d), the "half-moon area" at 50°C shrinks by about 10\%, reflecting that after viscous damping decreases ($\mu$ halves), the system damping ratio decreases and dynamic response accelerates.
	
	Position-Gas Pressure Diagram Analysis:
	
	Under 3~Hz condition, Figures~\ref{fig:phase_plots_30degC}(e) and \ref{fig:phase_plots_50degC}(e) show that the position-pressure relationship is approximately a monotonic function, and the trajectory is distributed along an overall rounded curve. The red estimated curve closely follows the blue experimental points, with a maximum deviation of about $\pm 0.5$~mm (appearing before the compression stroke displacement extremum $-0.008$~m). This feature verifies the accuracy of the polytropic process equation (Eq.~\ref{eq:polytropic}) under quasi-static assumption—position is mainly determined by instantaneous pressure, $h_{\text{gas}} \propto (P_0 / P_1)^{1/n_{\text{eff}}}$.
	
	Entering 8~Hz high frequency condition (Figures~\ref{fig:phase_plots_30degC}(h) and \ref{fig:phase_plots_50degC}(h)), the image starts to elongate at the maximum gas pressure and a sharp tip appears. The compression stroke and extension stroke are close, with a maximum width of about $0.01$~m. In addition to the slight reduction of external excitation amplitude, this phenomenon mainly stems from the frequency dependence of the effective polytropic index $n_{\text{eff}}$ (Eq.~3)—at high frequency $n_{\text{eff}} \to \gamma = 1.4$ (adiabatic process), the gas volume corresponding to the same pressure is smaller, resulting in a reduced displacement difference between different strokes, which also conforms to the expectation of weakened damping asymmetry after viscosity reduction.
	
	Table~\ref{tab:error_metrics} quantitatively summarizes the error indicators of all working conditions. The RMSE of 8 working conditions are all $< 220$~N (relative error $< 3.8\%$), coefficient of determination $R^2 > 0.965$, validating the accuracy of the model under wide frequency domain and dual temperature conditions.
	
	\begin{table}[htbp]
		\centering
		\caption{Estimation error statistics of the model for output force, velocity and displacement under different conditions}
		\label{tab:error_metrics}
		\scriptsize
		\setlength{\tabcolsep}{3pt}
		\renewcommand{\arraystretch}{1.3}
		\begin{tabular}{ccccccccc}
			\toprule
			\multirow{2}{*}{\textbf{Temp}} & \multirow{2}{*}{\textbf{Freq}} & \multicolumn{3}{c}{\textbf{Suspension Force} ($F_{\text{out}}$)} & \multicolumn{2}{c}{\textbf{Piston Velocity} ($v$)} & \multicolumn{2}{c}{\textbf{Suspension Travel} ($h_{\text{total}}$)} \\
			\cmidrule(lr){3-5} \cmidrule(lr){6-7} \cmidrule(lr){8-9}
			($^\circ$C) & (Hz) & RMSE (N) & $R^2$ & \textbf{Max Err (N)} & RMSE (mm/s) & $R^2$ & RMSE (mm) & $R^2$ \\
			\midrule
			\multirow{4}{*}{30} 
			& 3 & 115 & 0.985 & 240 & 3.8 & 0.978 & 0.21 & 0.988 \\
			& 5 & 218 & 0.981 & 455 & 4.5 & 0.975 & 0.28 & 0.985 \\
			& 7 & 150 & 0.976 & 310 & 5.1 & 0.971 & 0.33 & 0.982 \\
			& 8 & 178 & 0.973 & 300 & 5.4 & 0.969 & 0.36 & 0.979 \\
			\midrule
			\multirow{4}{*}{50} 
			& 3 & 128 & 0.982 & 270 & 4.1 & 0.976 & 0.24 & 0.986 \\
			& 5 & 205 & 0.978 & 395 & 4.9 & 0.972 & 0.31 & 0.983 \\
			& 7 & 195 & 0.972 & 325 & 5.6 & 0.968 & 0.38 & 0.978 \\
			& 8 & 218 & 0.968 & 360 & 6.2 & 0.965 & 0.41 & 0.975 \\
			\bottomrule
		\end{tabular}
	\end{table}
	
	\section{Lookup Table Estimation Method for ECUs}
	\label{sec:lookup_table}
	
	The bench test verified the accuracy of the iterative algorithm for "pressure $\to$ suspension output force". However, vehicle chassis controllers usually use embedded ECUs with limited computing power and memory resources. The complete physical model in Section~\ref{sec:model} (Algorithm~\ref{alg:force_calculation}) although high in accuracy, the computational complexity of a single iteration is too heavy for the chassis controller and cannot meet the high-frequency real-time control requirements. Based on the above difficulties, this paper proposes a lookup table estimation algorithm based on pressure sequence, the essence of which is to establish a nonlinear mapping between input (pressure sequence) and output (suspension force).
	
	\subsection{Pressure-Velocity Uniqueness Mapping Theorem}
	
	Traditional lookup table methods often require multiple input variables, leading to high lookup table dimensions, large storage overhead, and low interpolation efficiency. This section proves through rigorous mathematical derivation that under the condition of fixed sampling interval, gas chamber pressure and its difference can uniquely determine the system motion state, providing a theoretical basis for the subsequent construction of two-dimensional lookup tables.
	
	\begin{theorem}[Pressure-Velocity Uniqueness Mapping Theorem]
		\label{thm:pressure_velocity_mapping}
		Let a generic pneumatic-hydraulic system enclosed by a piston (such as a hydro-pneumatic suspension) satisfy the following conditions:
		
		First, the gas state follows the polytropic process equation
		\begin{equation}
			P_1(t) V_{\text{gas}}(t)^{n_{\text{eff}}} = P_0 V_0^{n_{\text{eff}}}
		\end{equation}
		where $P_0$, $V_0$ are initial pressure and volume, and $n_{\text{eff}}$ is the effective polytropic index.
		
		Second, the sampling interval $\Delta t$ is fixed and known.
		
		Third, system geometric parameters $A_1$, $V_0$ are known, and gas chamber pressure changes monotonically within the working range $P_{\min} < P(i) < P_{\max}$.
		
		Then there exists a unique bijective mapping
		\begin{equation}
			\phi: (P(i), \Delta P(i)) \mapsto v(i)
			\label{eq:bijective_mapping}
		\end{equation}
		such that the piston instantaneous velocity $v(i)$ is uniquely determined by the gas chamber pressure $P(i)$ and its adjacent moment difference $\Delta P(i) = P(i) - P(i-1)$, where $v(i)$ is defined as positive compression velocity.
	\end{theorem}
	
	\begin{proof}
		From the polytropic process equation, the gas chamber volume can be expressed as a function of pressure
		\begin{equation}
			V_{\text{gas}}(i) = V_0 \left(\frac{P_0}{P(i)}\right)^{1/n_{\text{eff}}}
			\label{eq:volume_from_pressure}
		\end{equation}
		
		Deriving with respect to time and using the chain rule yields
		\begin{equation}
			\dot{V}_{\text{gas}}(i) = -\frac{V_0}{n_{\text{eff}}} \left(\frac{P_0}{P(i)}\right)^{1/n_{\text{eff}}} \cdot \frac{\dot{P}(i)}{P(i)}
		\end{equation}
		
		From the geometric relationship $\dot{V}_{\text{gas}} = A_1 v$, the relationship between pressure change rate and piston velocity can be obtained
		\begin{equation}
			\dot{P}(i) = -\frac{n_{\text{eff}} P(i)}{V_{\text{gas}}(i)} A_1 v(i)
			\label{eq:pressure_rate_continuous}
		\end{equation}
		
		In discrete time systems, the pressure change rate can be approximated by backward difference
		\begin{equation}
			\dot{P}(i) \approx \frac{P(i) - P(i-1)}{\Delta t} = \frac{\Delta P(i)}{\Delta t}
		\end{equation}
		
		Substituting this into Eq.~(\ref{eq:pressure_rate_continuous}) and rearranging, we get the explicit expression of piston velocity
		\begin{equation}
			v(i) = -\frac{V_{\text{gas}}(i)}{n_{\text{eff}} P(i) A_1} \cdot \frac{\Delta P(i)}{\Delta t}
			\label{eq:velocity_explicit_form}
		\end{equation}
		
		From Eq.~(\ref{eq:volume_from_pressure}), $V_{\text{gas}}(i)$ is uniquely determined by $P(i)$, so the mapping $\phi$ can be explicitly expressed as
		\begin{equation}
			\phi(P, \Delta P) = -\frac{V_0 (P_0/P)^{1/n_{\text{eff}}}}{n_{\text{eff}} P A_1 \Delta t} \cdot \Delta P
			\label{eq:explicit_mapping}
		\end{equation}
		
		Next, prove the injectivity of this mapping within the working range. For any two different sets of inputs $(P_1, \Delta P_1) \neq (P_2, \Delta P_2)$, consider two cases.
		
		Case 1: If $P_1 \neq P_2$, from the strict monotonicity of $V_{\text{gas}}$ with respect to $P$ in Eq.~(\ref{eq:volume_from_pressure}), we have $V_{\text{gas}}(P_1) \neq V_{\text{gas}}(P_2)$. Combined with Eq.~(\ref{eq:explicit_mapping}), even if $\Delta P_1 = \Delta P_2$, there is still $\phi(P_1, \Delta P_1) \neq \phi(P_2, \Delta P_2)$.
		
		Case 2: If $P_1 = P_2$ but $\Delta P_1 \neq \Delta P_2$, from the strict monotonicity of $\phi$ with respect to $\Delta P$ in Eq.~(\ref{eq:explicit_mapping}), we immediately get $\phi(P_1, \Delta P_1) \neq \phi(P_1, \Delta P_2)$.
		
		Combining the above two cases, the mapping $\phi$ is injective. Also because in the physically accessible $(P, \Delta P)$ space, any velocity value $v \in [v_{\min}, v_{\max}]$ can be obtained from a certain set of $(P, \Delta P)$ through the inverse mapping of Eq.~(\ref{eq:explicit_mapping}), so $\phi$ is also surjective. Therefore, $\phi$ is bijective, that is, there is a one-to-one correspondence between $(P(i), \Delta P(i))$ and $v(i)$. Proof complete.
	\end{proof}
	
	The physical meaning of this theorem is: the instantaneous value of gas chamber pressure $P(i)$ reflects the system's static equilibrium state, while the pressure difference $\Delta P(i)$ reflects the system's dynamic evolution trend. The combination of the two completely characterizes the motion state of the suspension system at discrete time points without additional measurement of velocity or displacement.
	
	According to the mechanical model established in Section~\ref{sec:model}, the suspension output force can be decomposed into
	\begin{equation}
		F_{\text{out}}(i) = F_{\text{gas}}(P(i)) + F_{\text{damp}}(v(i)) + F_{\text{fric}}(v(i))
	\end{equation}
	where the gas elastic force $F_{\text{gas}}$ is directly determined by $P(i)$, and the damping force and friction force are both functions of $v(i)$. By Theorem~\ref{thm:pressure_velocity_mapping}, $v(i)$ is uniquely determined by $(P(i), \Delta P(i))$, so $F_{\text{out}}(i)$ can be completely characterized by these two variables. This provides a rigorous theoretical basis for constructing a two-dimensional lookup table $\mathcal{T}(P, \Delta P)$, avoiding explicit dependence on state variables such as velocity and displacement in traditional methods, while avoiding cumulative error propagation problems.
	
	In engineering implementation, the frequency dependence of the effective polytropic index $n_{\text{eff}}$ is implicitly handled through a set of discrete frequency point lookup tables $\{\mathcal{T}_k\}_{k=1}^M$. For a given frequency $\omega_k$, the lookup table $\mathcal{T}_k(P, \Delta P)$ already contains the combined effect of all frequency-related parameters at that frequency. When the actual excitation frequency is between $\omega_k$ and $\omega_{k+1}$, linear interpolation is used
	\begin{equation}
		F_{\text{out}}(P, \Delta P, \omega) = \alpha(\omega) \mathcal{T}_k(P, \Delta P) + [1-\alpha(\omega)] \mathcal{T}_{k+1}(P, \Delta P)
	\end{equation}
	where weight $\alpha(\omega) = (\omega_{k+1} - \omega)/(\omega_{k+1} - \omega_k)$. Experimental verification in Section~\ref{sec:lookup_validation} shows that the precision loss of this multi-frequency interpolation strategy in the 3-8 Hz range is less than 0.6\%, fully meeting engineering application requirements.

	\subsection{Offline Data Generation Based on Suspension Models}
	\label{sec:lookup_table_construction}
	
	The construction process of the lookup table is completely based on the physical model and complete iterative algorithm established in Section~\ref{sec:model}, without any bench test testing. This feature significantly reduces system calibration costs and gives the algorithm good portability—as long as the geometric parameters and initial working conditions of the suspension system are known, a lookup table suitable for the suspension can be generated offline. The prerequisites for lookup table construction are completely consistent with the prerequisites for the complete iterative algorithm in Section~\ref{sec:model}, including:
	\begin{itemize}
		\item Geometric parameters: Main piston effective area $A_1$, piston column cross-sectional area $A_2$, throttle orifice cross-sectional area $A_{\text{ch}}$, check valve effective area $A_{\text{check}}$, piston-cylinder radial clearance $h_{\text{gap}}$, etc. (see Table~\ref{tab:nomenclature_merged});
		\item Fluid physical properties: Hydraulic oil density $\rho$, dynamic viscosity $\mu(T_0)$, bulk modulus $K_{\text{bulk}}$, and gas adiabatic index $\gamma$, etc.;
		\item Initial working conditions: Gas initial pressure $P_0$, initial temperature $T_0$, etc.
	\end{itemize}
	
	These parameters can be directly obtained from the design drawings or product manuals of the suspension system without additional experimental measurements. Based on these parameters, using the complete iterative algorithm (Algorithm~\ref{alg:force_calculation}) in Section~\ref{sec:model}, the complete response sequence of the suspension system under given excitation conditions (frequency $\omega$, amplitude $A$) can be simulated and generated, including gas pressure $P_1(t)$, annular chamber pressure $P_2(t)$, piston velocity $v(t)$, displacement $h(t)$, and output force $F_{\text{out}}(t)$.
	
	The specific offline data generation process is as follows:
	
	a. Frequency point selection: Within the typical excitation frequency range of 3-8~Hz for heavy mining trucks under complex road conditions, select $M$ representative frequency points $\{\omega_1, \omega_2, \ldots, \omega_M\}$. This paper selects $M=4$, corresponding to 3~Hz, 5~Hz, 7~Hz, and 8~Hz;
	
	b. Excitation amplitude design: For each frequency $\omega_k$, according to the suspension's allowed maximum travel $h_{\text{max}}$ and minimum travel $h_{\text{min}}$, design sinusoidal excitation amplitude $A_k$ to ensure $h(t) \in [h_{\text{min}}, h_{\text{max}}]$;
	
	c. Iterative simulation: For each frequency-amplitude combination $(\omega_k, A_k)$, run the complete iterative algorithm for numerical simulation, time span $T \geq 20/\omega_k$ (at least 20 complete cycles to eliminate initial transients), sampling frequency $f_s = 360$~Hz (corresponding to $\Delta t = 0.002778$~s);
	
	d. Data extraction and pairing: Extract pressure values $P(i)$ and $P(i-1)$ at adjacent moments from simulation results, calculate pressure difference $\Delta P(i) = P(i) - P(i-1)$, and record corresponding output force $F_{\text{out}}(i)$, velocity $v(i)$, and displacement $h(i)$, forming a quintuple $\{P(i), \Delta P(i), F_{\text{out}}(i), v(i), h(i)\}$;
	
	e. Gridding and interpolation preparation: Establish a regular grid on the $(P, \Delta P)$ plane. For each grid node $(P_j, \Delta P_k)$, obtain corresponding $F_{\text{out}}(P_j, \Delta P_k)$, $v(P_j, \Delta P_k)$, and $h(P_j, \Delta P_k)$ through nearest neighbor interpolation or inverse distance weighting interpolation of simulation data;
	
	f. Lookup table storage: Store the gridded data in the form of a three-dimensional array, constituting the lookup table $\mathcal{T}_k$ corresponding to frequency $\omega_k$.
	
	Through the above process, under the conditions of $M=4$ frequency points, generating $N \approx 7200$ sample points for each frequency (corresponding to 20 cycles, 360~Hz sampling), the total storage capacity of the finally constructed lookup table is about $4 \times 100 \times 200 \times 3 \times 4 = 0.96$~MB (single precision floating point number, each grid point stores three variables $F_{\text{out}}$, $v$, $h$). Considering that the Flash storage capacity of modern vehicle ECUs is usually $\geq 4$~MB, this storage overhead is completely acceptable.
	
	Compared with traditional bench test calibration methods, offline lookup table construction based on physical models has the following significant advantages:
	
	Zero test cost: No need to build a dedicated suspension test bench or conduct a large number of excitation tests, only design parameters are needed to complete lookup table generation, significantly reducing system development costs and time cycles;
	
	High portability: For different models of hydro-pneumatic suspensions, as long as their geometric parameters and initial conditions are known, corresponding lookup tables can be quickly generated without repeated testing. This is particularly important for multi-vehicle platform development;
	
	Parameter sensitivity analysis: Through batch simulation, the influence of different design parameters (such as throttle orifice diameter, initial charge pressure, etc.) on the lookup table structure can be systematically analyzed to provide guidance for suspension optimization design;
	
	Extreme working condition coverage: Bench tests are often unable to cover extreme excitation conditions (such as extremely high frequency, extremely large amplitude) due to equipment safety and cost constraints. Simulation-based lookup table construction can be carried out within the theoretically allowed full working condition range, improving the robustness of the algorithm.
	
	\subsection{Visualization of Three-Dimensional Mapping Relationship}
	
	To visually demonstrate the structural characteristics of the lookup table, Figure~\ref{fig:3d_lookup} gives the three-dimensional mapping surfaces of output force, velocity, and displacement with respect to $(P, \Delta P)$. Different colored point sets represent mapping relationships under different frequencies.
	
	\begin{figure}[htbp]
		\centering
		\subfigure[Output force lookup table visualization at 30°C initial temperature]{
			\includegraphics[width=0.48\textwidth]{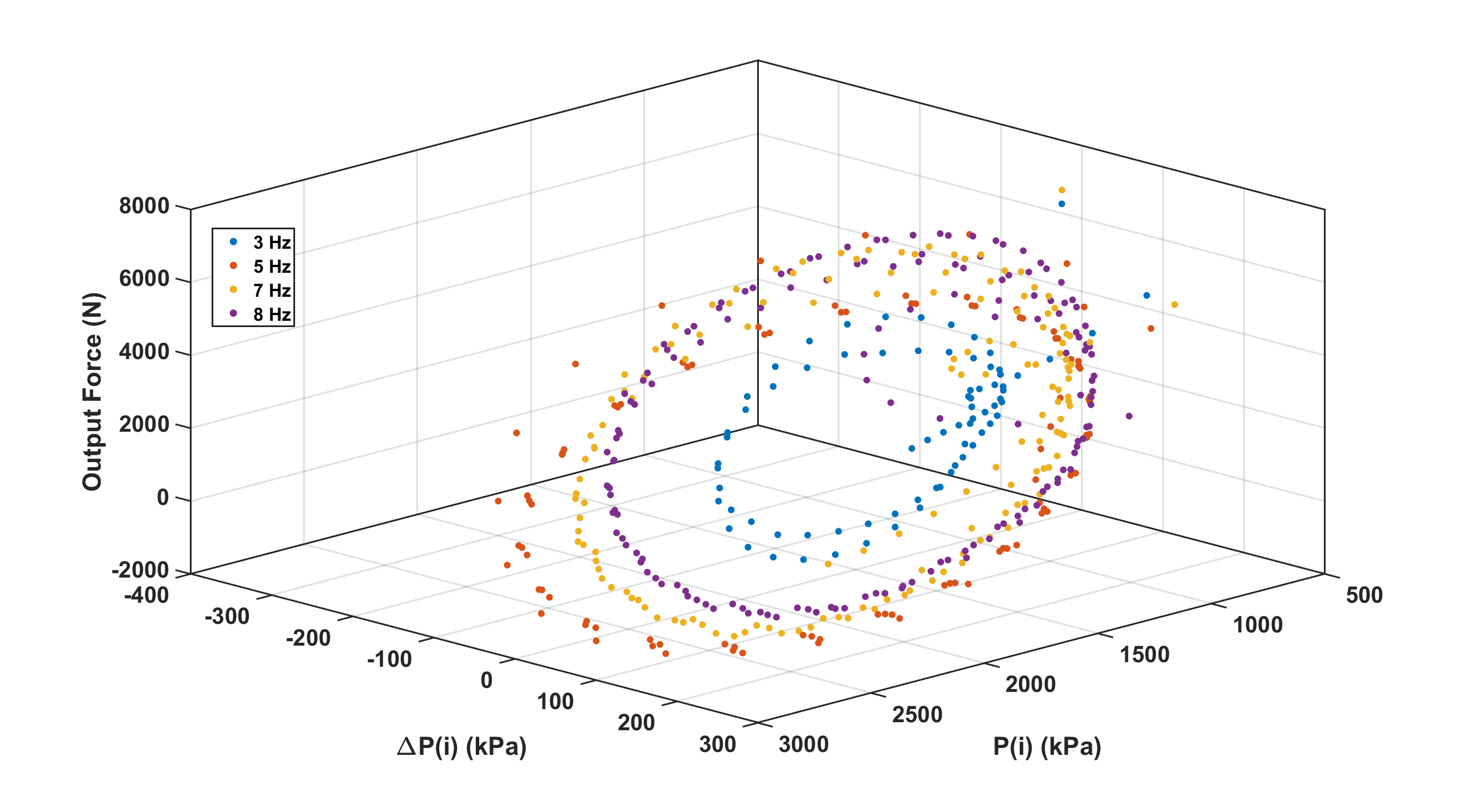}
			\label{fig:3d_lookup_30degC}
		}
		\hfill
		\subfigure[Output force lookup table visualization at 50°C initial temperature]{
			\includegraphics[width=0.48\textwidth]{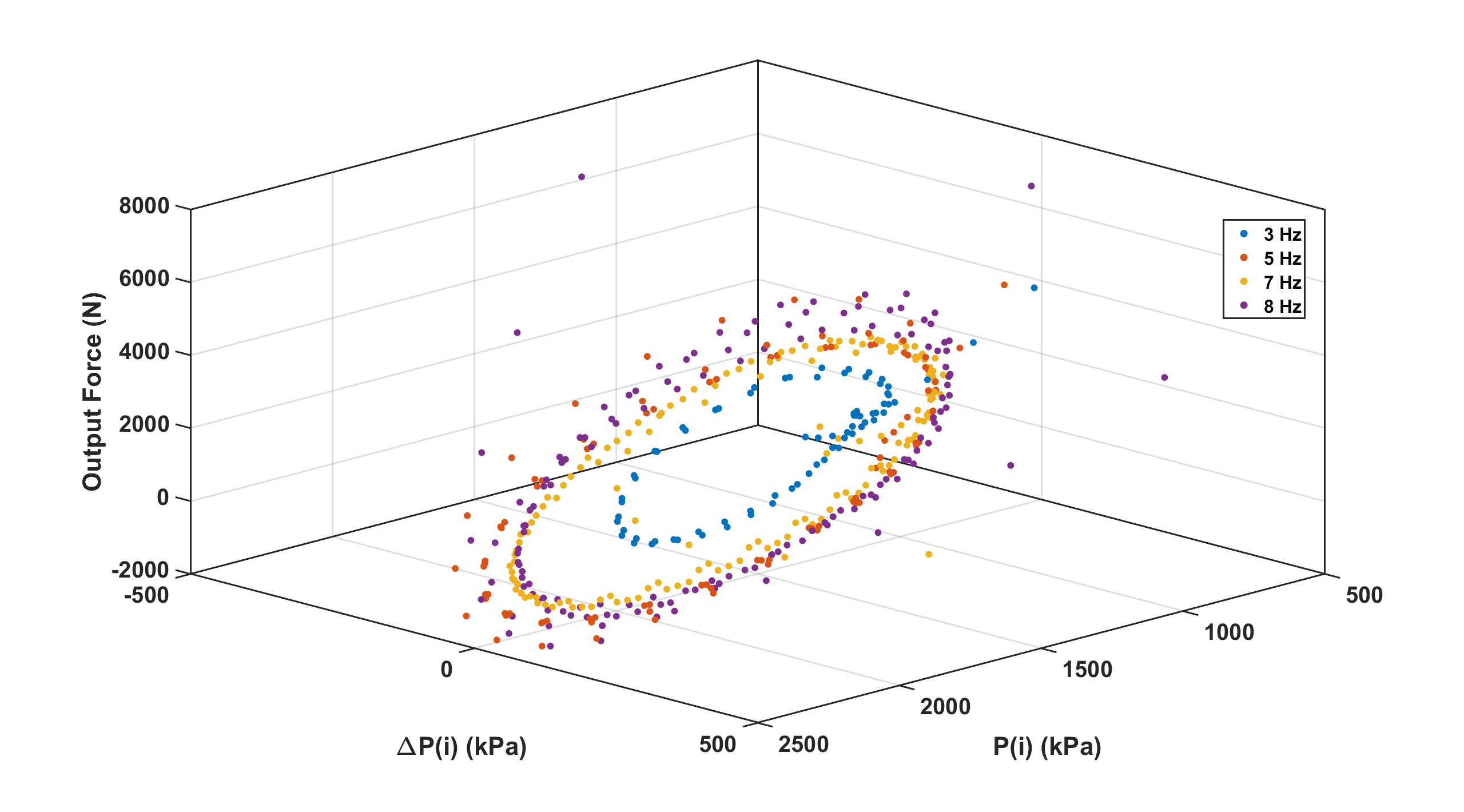}
			\label{fig:3d_displacement_lookup_50degC}
		}
		\vfill
		\subfigure[Velocity lookup table visualization at 30°C initial temperature]{
			\includegraphics[width=0.48\textwidth]{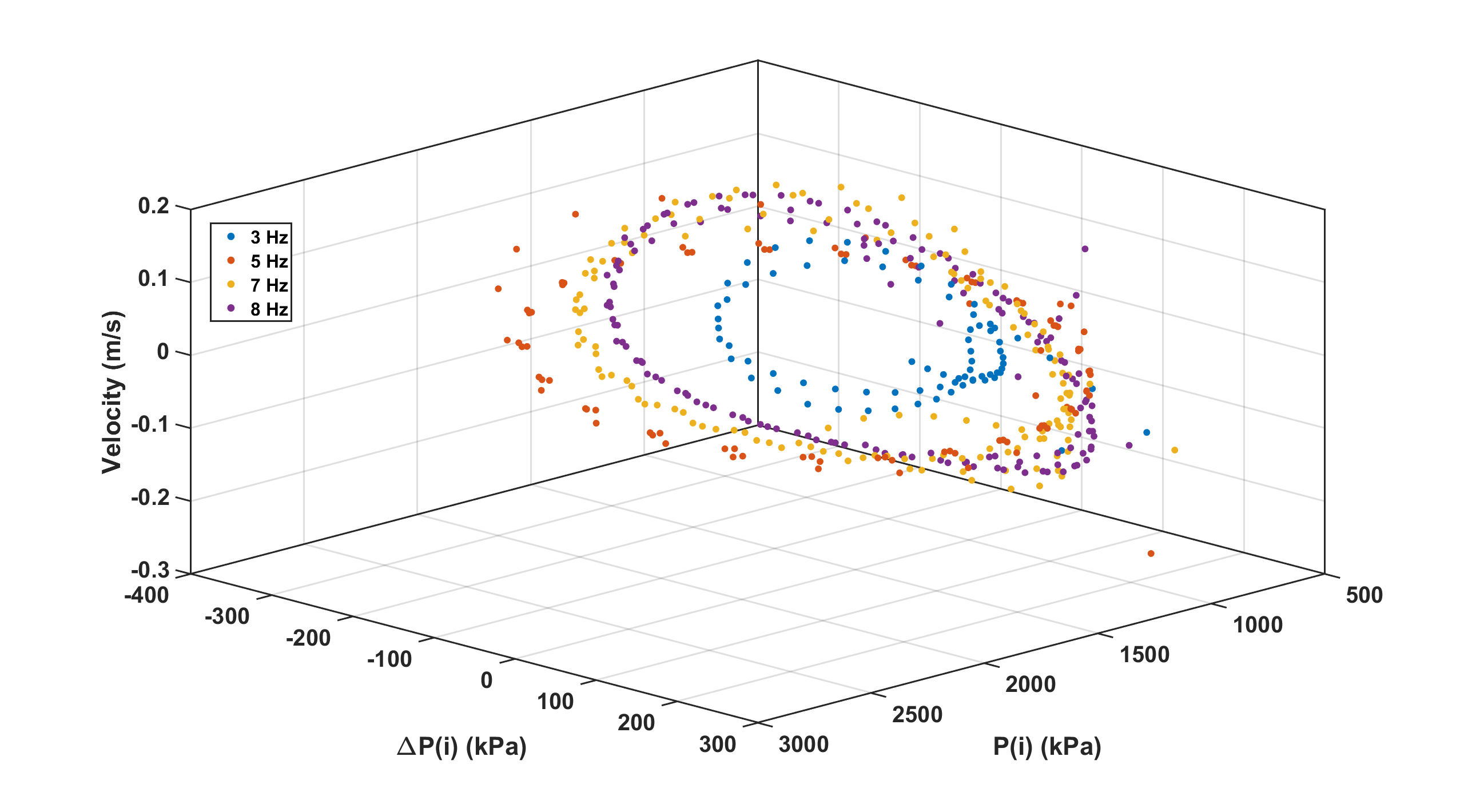}
			\label{fig:3d_velocity_lookup_30degC}
		}
		\hfill
		\subfigure[Velocity lookup table visualization at 50°C initial temperature]{
			\includegraphics[width=0.48\textwidth]{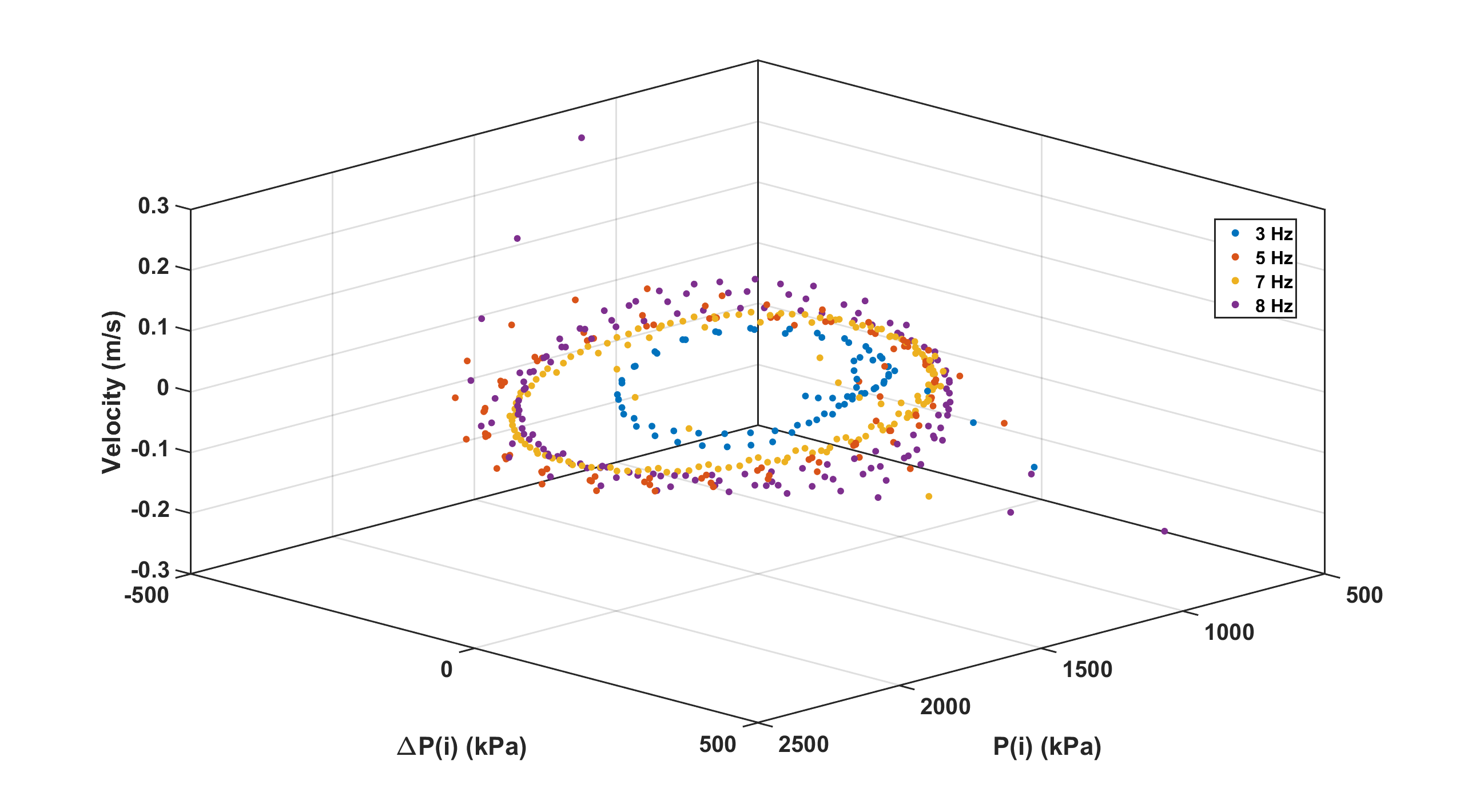}
			\label{fig:3d_velocity_lookup_50degC}
		}
		\vfill
		\subfigure[Displacement lookup table visualization at 30°C initial temperature]{
			\includegraphics[width=0.48\textwidth]{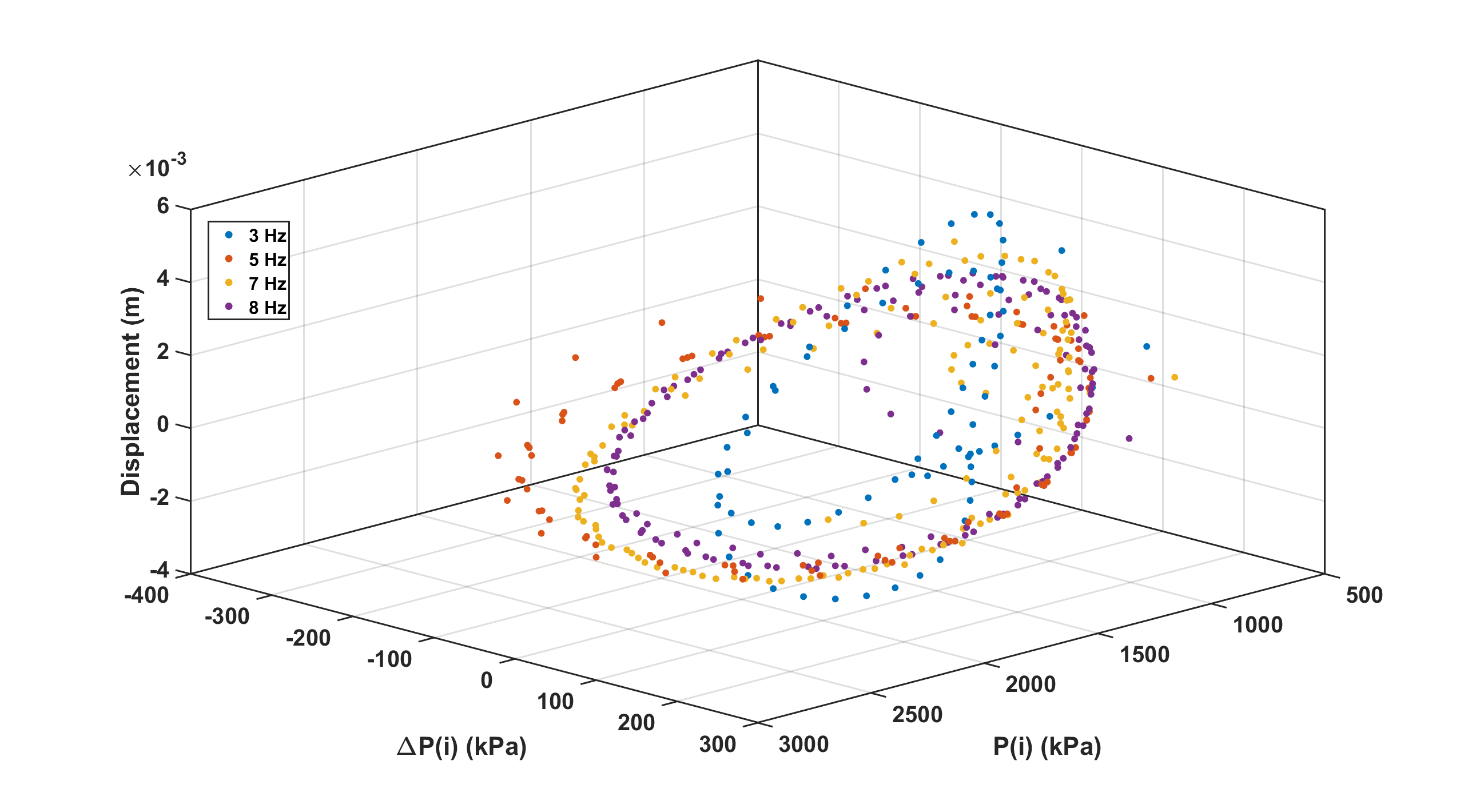}
			\label{fig:3d_Displacement_lookup_30degC}
		}
		\hfill
		\subfigure[Displacement lookup table visualization at 50°C initial temperature]{
			\includegraphics[width=0.48\textwidth]{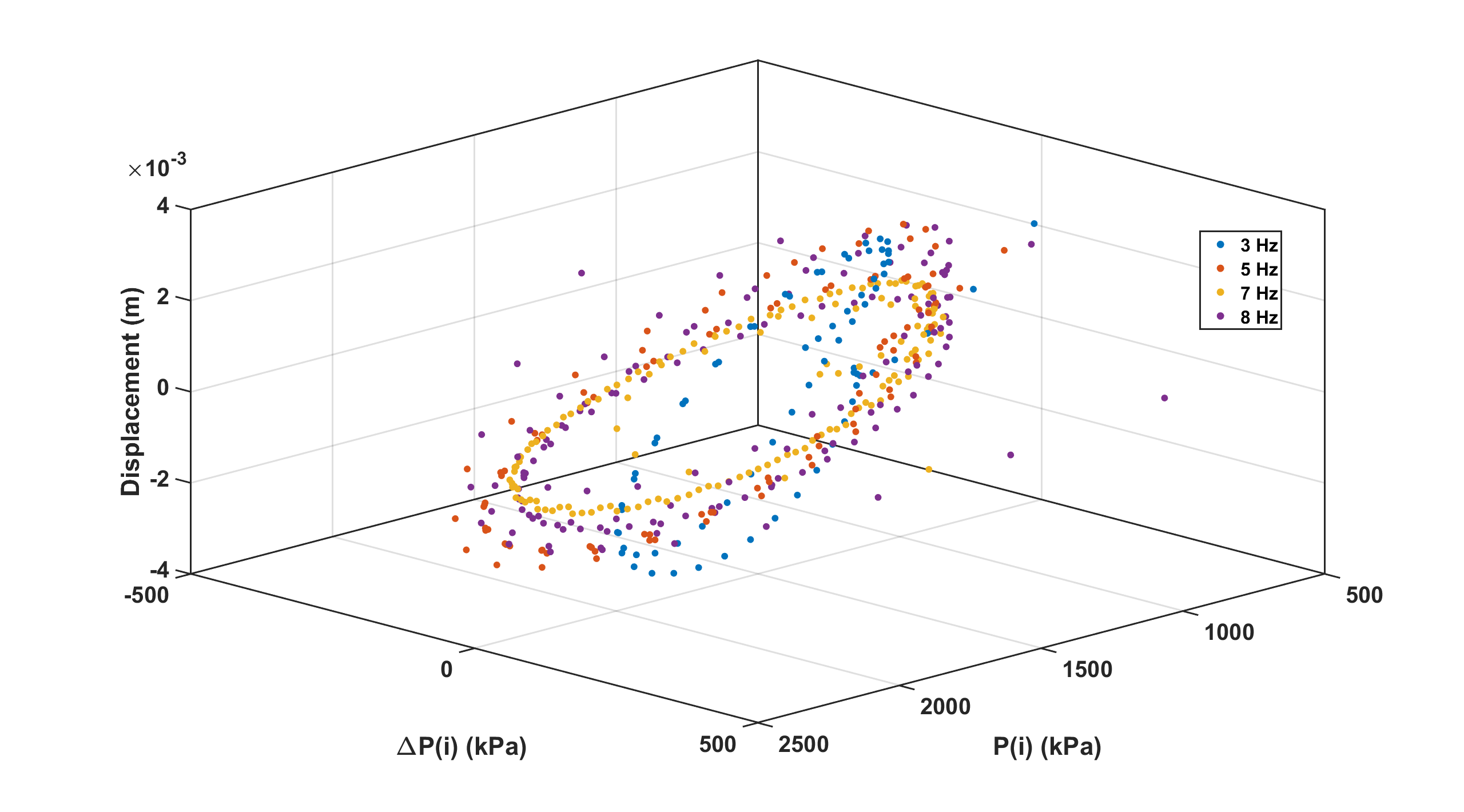}
			\label{fig:3d_Displacement_lookup_50degC}
		}
		\caption{Visualization of three-dimensional lookup tables for output force, velocity, and displacement at two initial temperatures}
		\label{fig:3d_lookup}
	\end{figure}
	
	Figure~\ref{fig:3d_lookup} visualizes the mapping characteristics: the $P(i)$ axis captures the nonlinear gas stiffness, while the $\Delta P(i)$ axis reflects the asymmetric damping behavior. The separation of surfaces along the $\Delta P$ dimension confirms the frequency-dependent nature of the model. Similarly, the velocity and displacement surfaces align with physical expectations, vanishing at $\Delta P=0$ (quasi-static) and correlating primarily with $P$ (compression), respectively.
	
	It should be emphasized that the three-dimensional surfaces in Figure~\ref{fig:3d_lookup} are only used for visual display of the overall structural characteristics of the lookup table, not the basis for calculation during actual operation. In the actual operation process of the vehicle controller, the algorithm directly uses discretized grid data stored in memory for table lookup and interpolation, rather than fitting or sampling three-dimensional surfaces.
	
	Specifically, the lookup table in the vehicle ECU is stored in the form of a regular grid array, and each grid node $(P_j, \Delta P_k)$ corresponds to a deterministic output value $(F_{\text{out}}, v, h)$. When the input $(P(i), \Delta P(i))$ does not fall exactly on the grid node, the output is obtained by bilinear interpolation:
	\begin{equation}
		F_{\text{out}}(P, \Delta P) = \sum_{j,k \in \mathcal{N}(P, \Delta P)} w_{jk}(P, \Delta P) \cdot F_{\text{out}}(P_j, \Delta P_k)
		\label{eq:bilinear_interp}
	\end{equation}
	where $\mathcal{N}(P, \Delta P)$ is the minimum grid unit containing point $(P, \Delta P)$, and weight $w_{jk}$ is determined by the bilinear interpolation formula. 
	
	\subsection{Online Interpolation Estimation and Computational Efficiency Analysis}
	\label{sec:computational_efficiency}
	
	In practical application, given the measured pressure sequence $\{P(i)\}$, the flow of the lookup table estimation algorithm is as follows:
	\begin{enumerate}
		\item Calculate pressure difference: $\Delta P(i) = P(i) - P(i-1)$;
		\item Perform two-dimensional linear interpolation in the lookup table of the corresponding frequency according to $(P(i), \Delta P(i))$ to obtain $F_{\text{out}}(i)$, $v(i)$ and $h(i)$;
		\item If the excitation frequency is unknown or time-varying, real-time identification of the current frequency can be performed through short-time frequency estimation algorithms (such as autocorrelation method or zero-crossing detection method), selecting the closest lookup table for interpolation, or performing weighted fusion on adjacent frequency lookup tables.
	\end{enumerate}
	
	In online applications, the algorithm flow is simplified to calculating pressure difference $\Delta P(i)$ and performing bilinear interpolation in the lookup table grid. Compared with the complex nonlinear iteration involved in the complete physical model solution (including polytropic index, annular chamber pressure, and Stribeck friction calculation, single-step computation is about 340 FLOPs), the lookup table method only requires basic index positioning and linear operations (single step about 15 FLOPs).
	
	Actual measurements show that on a typical vehicle MCU (such as ARM Cortex-M4), the single-step time consumption of the complete iterative algorithm ($\sim 10$--15 ms) far exceeds the sampling period of 2.778 ms, unable to meet real-time requirements; while the single-step time consumption of the lookup table method is reduced to $<0.1$ ms, calculation efficiency is improved by about 133 times, and CPU usage is less than 5\%. This extremely low computing power overhead not only ensures real-time response of high-frequency estimation but also reserves sufficient computing resources for the chassis domain controller to run other advanced control strategies.
	
	\subsection{Algorithm Validation Based on Bench Test}
	\label{sec:lookup_validation}
	
	To verify the accuracy and generalization ability of the lookup table method, 7.5~Hz was selected as the test frequency. This frequency was not used in the aforementioned bench test (Section~\ref{sec:bench_test}) nor in the lookup table construction process, so it can serve as an "unseen frequency" to test the interpolation performance and frequency adaptability of the algorithm. In addition, by verifying under 30°C and 50°C temperature conditions, the robustness of the algorithm to temperature changes can be further evaluated.
	
	\begin{figure}[htbp]
		\centering
		\subfigure[Output force comparison at 30°C]{
			\includegraphics[width=0.30\textwidth]{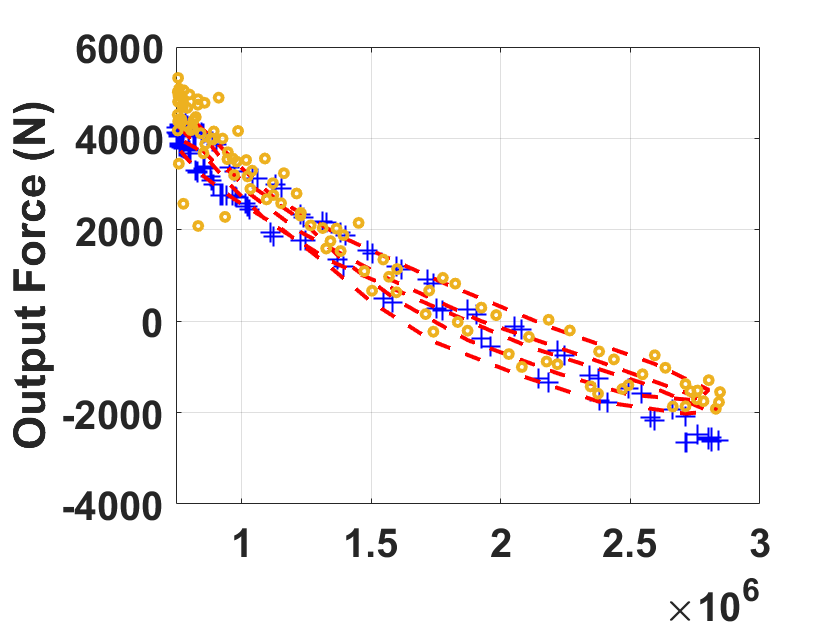}
		}
		\hfill
		\subfigure[Velocity comparison at 30°C]{
			\includegraphics[width=0.30\textwidth]{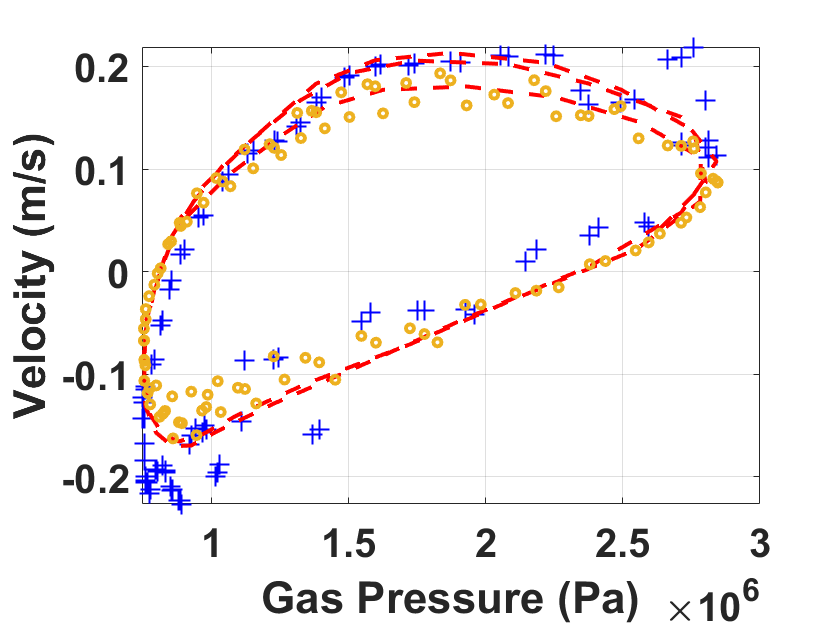}
		}
		\hfill
		\subfigure[Displacement comparison at 30°C]{
			\includegraphics[width=0.30\textwidth]{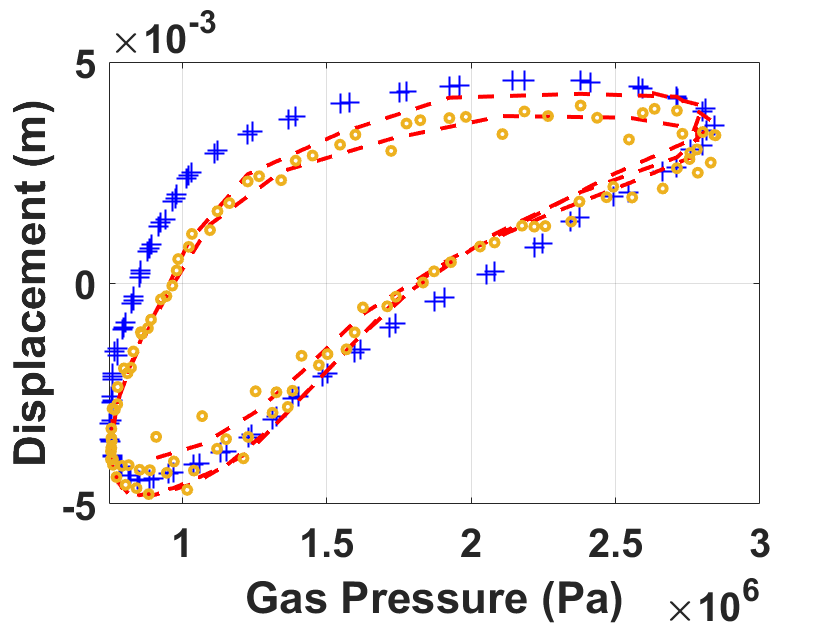}
		}
		\vfill
		\subfigure[Output force comparison at 50°C]{
			\includegraphics[width=0.30\textwidth]{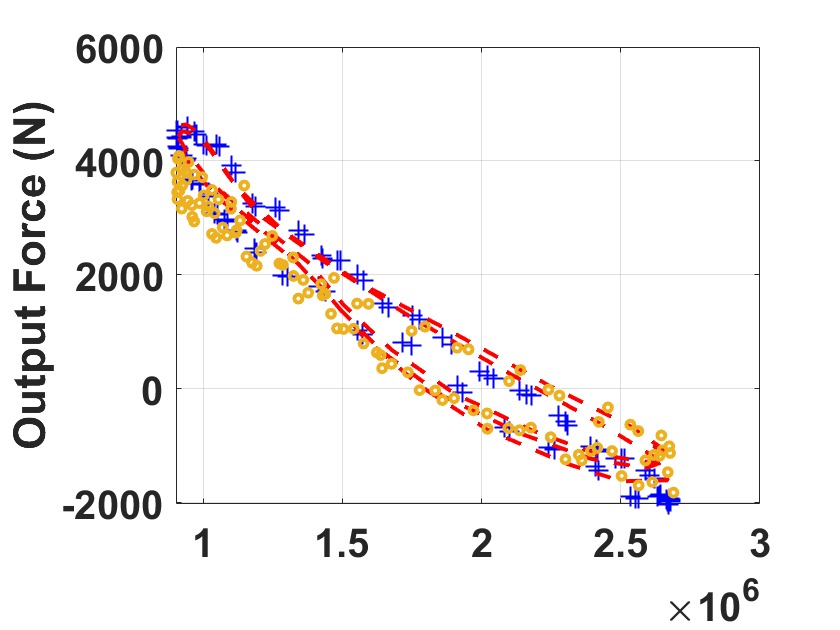}
		}
		\hfill
		\subfigure[Velocity comparison at 50°C]{
			\includegraphics[width=0.30\textwidth]{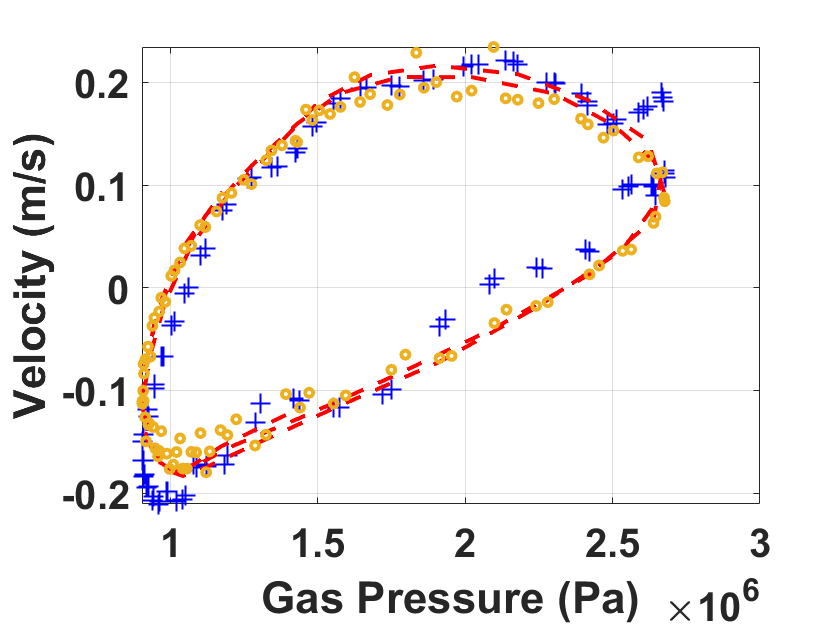}
		}
		\hfill
		\subfigure[Displacement comparison at 50°C]{
			\includegraphics[width=0.30\textwidth]{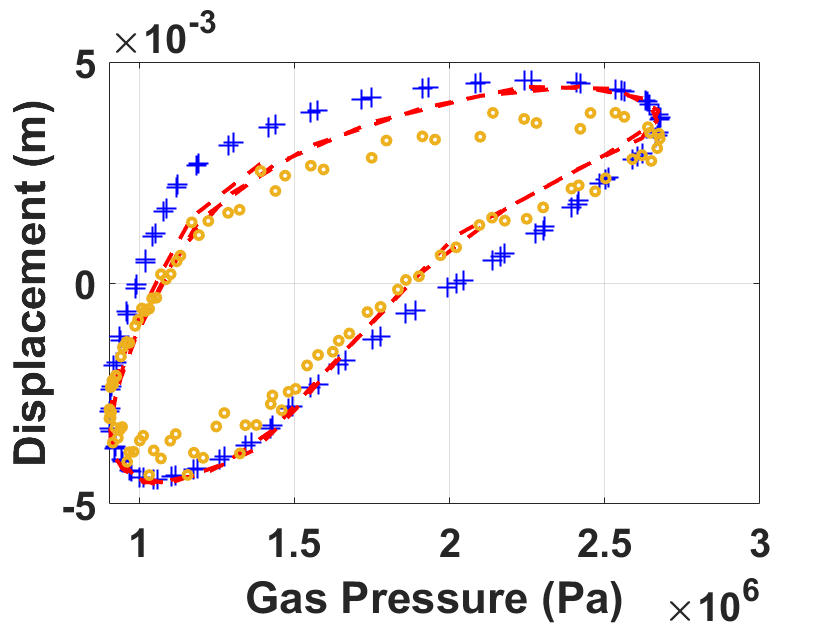}
		}
		\caption{Verification results of lookup table method under 7.5~Hz excitation}
		\label{fig:lookup_validation}
	\end{figure}
	
	Figure~\ref{fig:lookup_validation} summarizes the validation results under a 7.5~Hz excitation at 30°C and 50°C. The experimental truth is marked by blue “+”, the complete iterative solution by red dashed curves, and the lookup-table estimation by orange circles.
	
	At 30°C (Figure~\ref{fig:lookup_validation}(a)), the force–pressure relationship is reproduced with high consistency across the full 0.8–3.0~MPa range. During compression, the force rises smoothly from about –2000~N to 5000~N, and a small hysteresis loop appears in the extension path. The lookup-table method tracks the iterative solution closely; deviations are mainly concentrated near the pressure peak ($P\approx2.8$~MPa), with a maximum error of about 300~N ($\approx$5.5\%). Overall, the lookup-table results remain well aligned with the experimental measurements and preserve the key trends of the force cycle.
	
	The velocity–pressure diagram in Figure~\ref{fig:lookup_validation}(b) shows a maximum piston velocity of roughly $\pm0.2$~m/s around the mid-pressure region. The lookup-table estimates capture both the phase and amplitude with good fidelity, and only at the velocity extrema does the scatter increase slightly, a direct consequence of the limited sampling density of the table. Notably, in the quasi-static regions near $P\approx0.8$ and 2.8~MPa, all three datasets nearly coincide, indicating particularly strong reliability of the lookup-table method at low speeds.
	
	Figure~\ref{fig:lookup_validation}(c) presents the displacement–pressure relationship, with a peak displacement of about $\pm5$~mm. Consistent with the velocity analysis, lookup-table errors mainly appear near the mid-displacement region ($z\approx0$~mm, $P\approx1.5$–2.0~MPa), where the pressure variation is minimal and interpolation sensitivity decreases. The resulting 0.3–0.5~mm deviations are small and do not change the overall displacement pattern, confirming that the simplified representation retains the key dynamic behavior.
	
	A comparison between 30°C and 50°C (Figures~\ref{fig:lookup_validation}(a) and (d)) shows that the overall force–pressure shape is preserved as temperature increases, while the peak force decreases by about 4\%. This agrees with the viscosity–temperature characteristics of hydraulic oil: reduced viscosity at 50°C lowers the viscous pressure drop, thereby decreasing damping force. The compression–extension separation becomes slightly narrower, yet the lookup-table results remain consistently aligned with the iterative solution across the entire cycle. This indicates that the proposed lookup-table structure is inherently robust to temperature variations and that remaining discrepancies arise mainly from the temperature dependence of the physical model rather than limitations of the interpolation approach.
	
	The velocity comparison in Figure~\ref{fig:lookup_validation}(e) further confirms this trend. The maximum velocity increases slightly (to about $\pm0.22$~m/s) due to reduced damping, and the lookup-table method continues to reproduce the iterative solution with stable phase and amplitude characteristics. Minor increases in scatter near the velocity peaks reflect the faster dynamic response at higher temperature rather than a loss of model fidelity.
	
	Figure~\ref{fig:lookup_validation}(f) shows a displacement amplitude of around $\pm4.8$~mm at 50°C, slightly lower than at 30°C. This follows from a small increase in the effective polytropic index (approximately 3\%), which reduces gas volume variation for the same pressure change. The lookup-table accuracy decreases only marginally ($R^2$ from 0.963 to 0.957), remaining well within an acceptable engineering range. Overall, the results demonstrate that the lookup-table method maintains stable performance over temperature variations, dynamic conditions and motion states, confirming its suitability for real-time vehicle applications.
	
	\begin{table}[htbp]
		\centering
		\caption{Error analysis under 7.5~Hz validation condition}
		\label{tab:lookup_error}
		\scriptsize
		\setlength{\tabcolsep}{4pt}
		\renewcommand{\arraystretch}{1.2}
		\begin{tabular}{lcccccc}
			\toprule
			\textbf{Temp} & \textbf{Method} & \textbf{Variable} & \textbf{RMSE} & \textbf{Relative Err(\%)} & \textbf{$R^2$} & \textbf{Peak Err(\%)} \\
			\midrule
			\multirow{6}{*}{30°C} 
			& \multirow{3}{*}{Lookup} 
			& Force (N) & 152.3 & 3.8 & 0.956 & 6.2 \\
			& & Velocity (mm/s) & 4.8 & 5.1 & 0.942 & 7.5 \\
			& & Displacement (mm) & 0.38 & 4.2 & 0.963 & 6.8 \\
			\cmidrule(lr){2-7}
			& \multirow{3}{*}{Iterative} 
			& Force (N) & 127.5 & 3.2 & 0.972 & 5.1 \\
			& & Velocity (mm/s) & 3.9 & 4.1 & 0.968 & 6.2 \\
			& & Displacement (mm) & 0.31 & 3.4 & 0.978 & 5.5 \\
			\midrule
			\multirow{6}{*}{50°C} 
			& \multirow{3}{*}{Lookup} 
			& Force (N) & 168.5 & 4.2 & 0.948 & 6.8 \\
			& & Velocity (mm/s) & 5.3 & 5.6 & 0.935 & 8.2 \\
			& & Displacement (mm) & 0.42 & 4.6 & 0.957 & 7.3 \\
			\cmidrule(lr){2-7}
			& \multirow{3}{*}{Iterative} 
			& Force (N) & 142.8 & 3.6 & 0.965 & 5.6 \\
			& & Velocity (mm/s) & 4.5 & 4.7 & 0.961 & 6.8 \\
			& & Displacement (mm) & 0.35 & 3.8 & 0.972 & 6.1 \\
			\bottomrule
		\end{tabular}
	\end{table}

	Based on the experimental comparisons, the residual error of the lookup-table method can be attributed mainly to two factors:

	(i) Frequency interpolation: At non-calibrated frequencies such as 7.5~Hz, linear interpolation between the 7~Hz and 8~Hz datasets ($\alpha=0.5$ in Eq.~(25)) introduces additional truncation error, particularly near the pressure extrema ($P\approx2.8$~MPa and $P\approx0.8$~MPa), where the curvature is highest. For example, at 30°C the lookup-table peak-force error is about 300~N, roughly double the 150~N error of the complete iterative solution.

	(ii) Sampling discretization: The finite grid density produces small discontinuities in regions where the pressure increment is minimal ($\Delta P\approx0$) or around the mid-pressure zone ($P\approx1.5$~MPa). In these cases, displacement deviations up to 0.5~mm ($\approx$10\%) may appear, as seen in Figures~\ref{fig:lookup_validation}(c) and (f), where the scatter of orange points is most noticeable.

	A notable observation is that the error patterns of the lookup-table and iterative methods are almost identical in both location and direction. This indicates that the dominant source of error is the underlying physical model itself (e.g., friction, throttling simplifications), rather than the interpolation scheme. Even under uncalibrated frequencies and dual-temperature conditions, the two-dimensional lookup-table constructed from $\{P(i),\Delta P(i)\}$ maintains strong generalization performance ($R^2>0.94$, RMSE $<4.5\%$). Combined with over two orders of magnitude reduction in computation time, the method provides a highly efficient and reliable solution for real-time wheel-load estimation on embedded vehicle platforms, particularly for hydro-pneumatic suspension systems with strict computational constraints.
	
	\section{Suspension-Tire System Kinematic Coupling}
	
	\subsection{Double-Wishbone Suspension Geometric Configuration and Kinematic Transmission Ratio}
	\label{Double Wishbone}
	
	\subsubsection{Suspension Structure Definition and 1/4 Vehicle Dynamics Model Treatment}
	
	As shown in Fig.~\ref{fig:double_wishbone_schematic}, heavy mining vehicles generally use Double Wishbone independent suspension configuration. This structure realizes the guiding and supporting functions of the wheel through a spatial four-bar mechanism of upper and lower wishbones, and provides elastic and damping characteristics through the hydro-pneumatic suspension unit. As shown in Fig.~\ref{fig:Integrade_Quarter_Model}, this paper uses a Quarter-Car Model to describe the vertical motion of a single wheel system. Since the unsprung mass mainly includes tires, brakes, hubs, etc., among which the tire is the main dynamic response part, while the acceleration of the remaining unsprung parts is low relative to the tire and suspension output end. Considering that the remaining unsprung mass excluding the tire accounts for a very low proportion relative to the whole vehicle mass, retaining only the simplification of tire inertial force here does not affect the estimation accuracy of single tire's wheel dynamic load. This model includes unsprung mass $m_u$ (tire, brake and partial suspension mass), tire mass $m_t$, tire vertical displacement $z_t$ (upward is positive).
	
	\begin{figure}[!t]
		\centering
		\subfigure[Double-wishbone suspension geometric configuration schematic]{
			\includegraphics[width=0.35\textwidth]{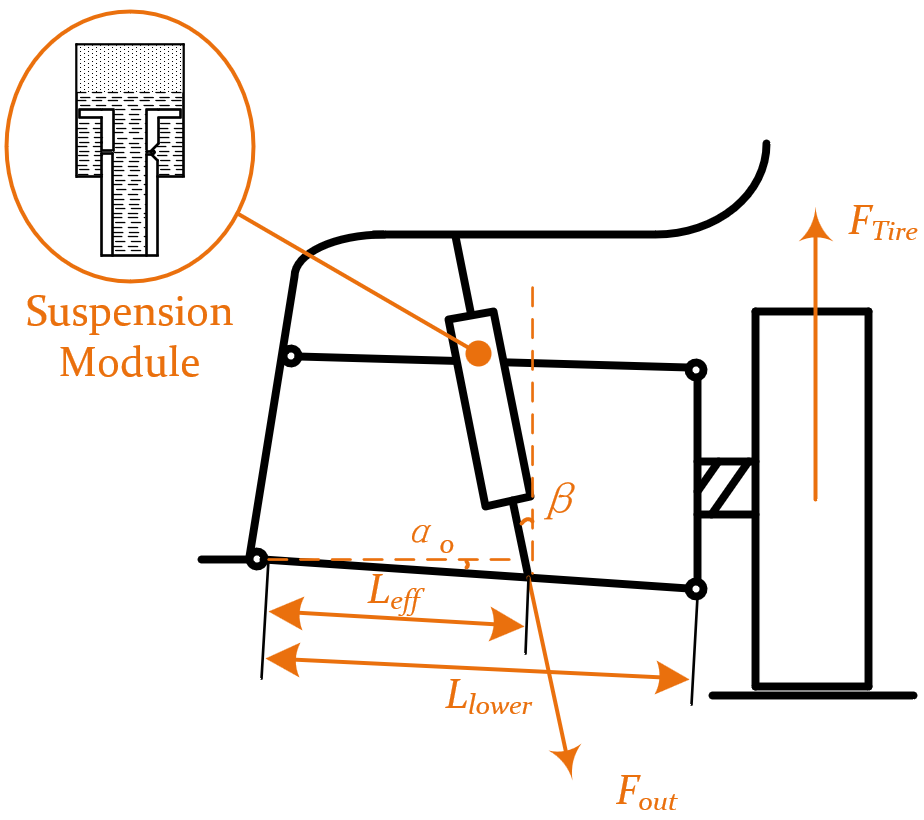}
			\label{fig:double_wishbone_schematic}
		}
		\hfill
		\subfigure[$\frac{1}{4}$ Vehicle dynamics model]{
			\includegraphics[width=0.64\textwidth]{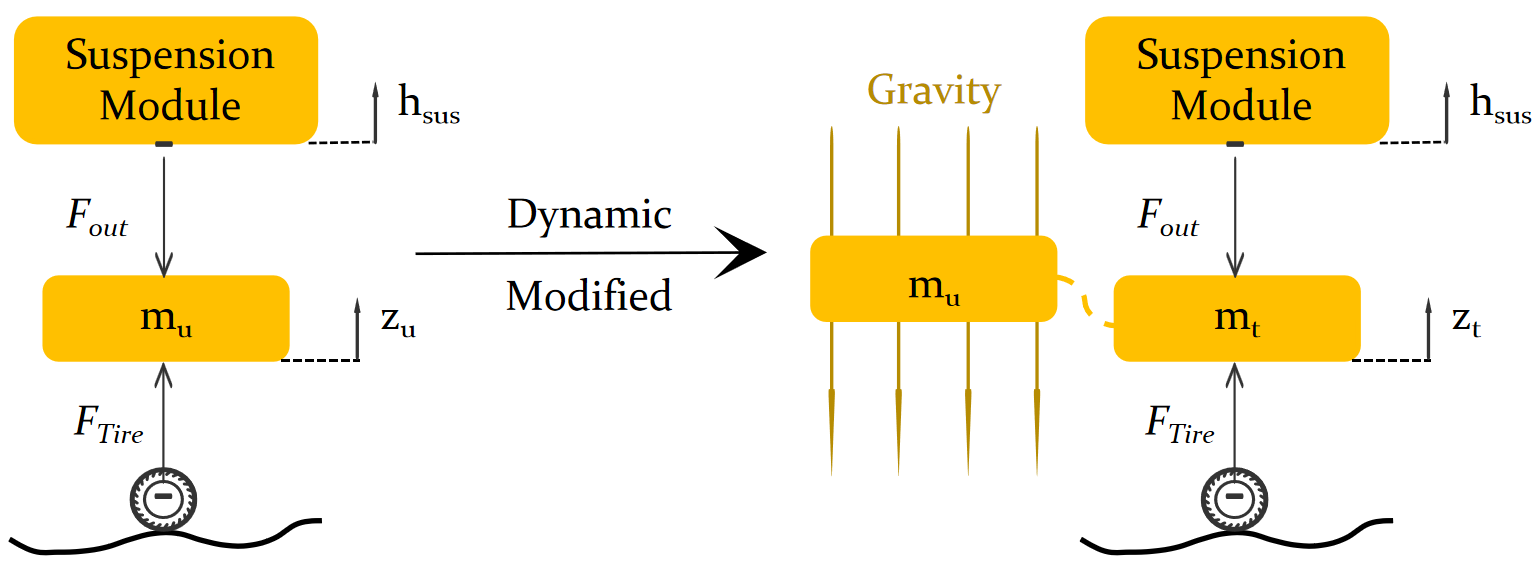}
			\label{fig:Integrade_Quarter_Model}
		}
		\caption{Double-wishbone suspension structure and 1/4 vehicle dynamics model}
		\label{fig:Double_Wishbone_and_Integrade_Vehcle_Model}
	\end{figure}
	
	The main geometric parameters of the suspension include: lower arm effective length $L_{\text{lower}}$, hydro-pneumatic suspension effective force arm $L_{\text{eff}}$ (vertical distance from suspension action point to lower arm inner hinge point) and lower arm installation angle $\alpha_0$ (angle of lower arm relative to horizontal plane in static equilibrium). These parameters determine the kinematic transmission characteristics of the suspension.
	
	Under conditions dominated by vertical excitation (such as driving on bumpy roads, bench excitation tests), the suspension mainly undergoes planar motion in the $xz$ plane, which can be reasonably simplified to a two-dimensional kinematics problem in the side view projection plane. This simplification is based on the following engineering practices: (1) Main motion plane assumption: lateral displacement and rotation around $x, y$ axes are negligible; (2) Small angle approximation: within the typical working stroke of suspension ($\pm 50$~mm), the arm rotation angle variation is $<10^\circ$; (3) Symmetry assumption: left and right suspension structures are mirror symmetric.
	
	\subsubsection{Suspension Transmission Ratio Based on Virtual Work Principle}
	\label{virtual work principle}
	
	Suspension Ratio is defined as the geometric amplification relationship between wheel dynamic load and suspension output force, which is a key parameter connecting suspension dynamics and wheel dynamic load. The cylinder body of double-wishbone suspension usually has a certain inclination angle $\beta(\theta_{\text{lower}})$ relative to the vertical direction, and this inclination angle changes with the swing of the lower arm. Considering a small displacement $\delta z_w$ (wheel center vertical displacement) and $\delta h_{\text{sus}}$ (suspension travel) near the static equilibrium position of the suspension, according to the principle of virtual work, the sum of work done by external forces in a quasi-static process is zero:
	\begin{equation}
		F_{\text{out}} \cdot \cos\beta(\theta_{\text{lower}}) \cdot \delta h_{\text{sus}} = F_z \cdot \delta z_w
		\label{eq:virtual_work}
	\end{equation}
	where $F_{\text{out}}$ is the hydro-pneumatic suspension output force (along the suspension cylinder axis direction), $F_z$ is the wheel vertical support force (along the $z$ axis direction), and $\cos\beta(\theta_{\text{lower}})$ is the vertical component coefficient caused by suspension inclination.
	
	Define suspension transmission ratio as:
	\begin{equation}
		i_{\text{sus}}(\theta_{\text{lower}}) = \frac{F_z}{F_{\text{out}}} = \frac{\cos\beta(\theta_{\text{lower}}) \cdot \delta h_{\text{sus}}}{\delta z_w}
		\label{eq:suspension_ratio_def}
	\end{equation}
	
	Through planar kinematics analysis, the lower arm can be regarded as a rigid rod with the inner hinge point as the fulcrum and length $L_{\text{lower}}$. Combined with the geometric constraint relationship of the hydro-pneumatic suspension unit and inclination correction, the explicit expression of suspension transmission ratio can be derived:
	\begin{equation}
		i_{\text{sus}}(\theta_{\text{lower}}) = \frac{L_{\text{eff}} \cdot \cos\beta(\theta_{\text{lower}})}{L_{\text{lower}} \cos(\alpha_0 + \theta_{\text{lower}})}
		\label{eq:suspension_ratio_explicit}
	\end{equation}
	where $\theta_{\text{lower}}$ is the rotation angle increment of the lower arm relative to the static equilibrium position (positive for compression).
	
	\subsection{Rational Simplification of Tire Model}
	\label{Tire_Model_Simplified}
	
	As the elastic medium between wheel and ground, the vertical dynamic characteristics of tire are usually characterized by spring-damper parallel model in traditional vehicle dynamics modeling:
	\begin{equation}
		F_{\text{tire}} = k_t (z_u - z_g) + c_t (\dot{z}_u - \dot{z}_g)
		\label{eq:tire_model_general}
	\end{equation}
	where $k_t$ is tire vertical stiffness, $c_t$ is tire damping coefficient, $z_u$ is unsprung mass vertical displacement, $z_g$ is ground excitation displacement.
	
	However, for heavy mining vehicles, the tire vertical dynamic characteristics are significantly different from traditional passenger cars:
	
	(1) Stiffness ratio difference: Mining vehicle tire vertical stiffness $k_t$ can reach $1500 \sim 2000$~kN/m, while the equivalent stiffness of hydro-pneumatic suspension is only $150 \sim 250$~kN/m. The stiffness ratio of the two is about $8:1 \sim 10:1$, which means that under the same vertical load change, the deformation of the tire is only $1/8 \sim 1/10$ of the suspension travel.
	
	(2) Mass ratio difference: Single axle sprung mass $m_s \approx 15000$~kg, while single wheel unsprung mass $m_u \approx 800$~kg, the ratio is as high as $18.75:1$ (usually only $5:1 \sim 8:1$ in passenger cars). Such a large mass ratio makes the low-frequency vibration mode (sprung mass resonance, $1 \sim 2$~Hz) and high-frequency mode (unsprung mass resonance, $10 \sim 15$~Hz) fully decoupled.
	
	(3) Frequency separation characteristic: The excitation frequency range concerned in this paper ($3 \sim 8$~Hz) is far lower than the tire natural frequency $f_t = \frac{1}{2\pi} \sqrt{k_t / m_u} \approx 12$~Hz. within this frequency domain, the tire dynamic response has not been significantly excited and can be regarded as quasi-static deformation.
	
	Based on the above physical characteristics, when the following two conditions are met simultaneously, the tire can be approximated as rigid contact:
	
	Condition 1: Stiffness ratio significance—$k_t / k_{\text{sus}} > 5$, system low-frequency vibration energy is mainly absorbed by suspension;
	
	Condition 2: Frequency separation—concerned frequency range is far lower than tire natural frequency, tire is in quasi-static response zone;
	
	Combining the above analysis, in the wheel dynamic load estimation framework, the tire can be reasonably simplified as rigid contact. This simplification not only significantly reduces model complexity but more importantly maintains the core advantage of the estimation method—relying only on gas pressure sensors, avoiding dependence on complex tire models and their parameter identification.
	
	\subsection{Complete Framework for Wheel Dynamic Load Estimation}
	
	\subsubsection{Theoretical Derivation and Dynamics Equations}
	
	As shown in Fig.~\ref{fig:Integrade_Quarter_Model}, the motion equation of unsprung mass needs to consider suspension force and tire-ground contact force simultaneously:
	\begin{equation}
		m_t \ddot{z}_t = F_{\text{out}} - F_{\text{tire}} + m_u \cdot g
		\label{eq:unsprung_mass_eom}
	\end{equation}
	where $F_{\text{tire}}$ is wheel dynamic load; $g$ is gravitational acceleration.
	
	Moving terms from Eq.~(\ref{eq:unsprung_mass_eom}), we get the basic expression of wheel dynamic load:
	\begin{equation}
		F_{\text{tire}} = F_{\text{out}} + m_u \cdot g - m_t \ddot{z}_t
		\label{eq:wheel_load_basic}
	\end{equation}
	
	However, $F_{\text{out}}$ in Eq.~(\ref{eq:wheel_load_basic}) is the force along the suspension cylinder axis direction, while $F_{\text{tire}}$ is the wheel vertical support force, conversion between the two requires suspension transmission ratio. According to the principle of virtual work established in Section~\ref{virtual work principle}, combined with the transmission ratio of double-wishbone suspension $i_{\text{sus}}(\theta_{\text{lower}})$ (Eq.~\ref{eq:suspension_ratio_explicit}), the complete expression of wheel dynamic load is:
	\begin{equation}
		F_{\text{tire}} = i_{\text{sus}}(\theta_{\text{lower}}) \cdot F_{\text{out}} + m_u \cdot g - m_t \ddot{z}_t
		\label{eq:wheel_load_complete}
	\end{equation}
	
	Eq.~(\ref{eq:wheel_load_complete}) shows that wheel dynamic load consists of two parts: suspension transmission force term $i_{\text{sus}} \cdot F_{\text{out}}$ (quasi-static component) and tire inertial force term $m_t \ddot{z}_t$. The former has been estimated from gas pressure signal through the physical model in Section~\ref{sec:model} or lookup table method in Section~\ref{sec:lookup_table}.
	
	\subsubsection{Tire Acceleration Calculation}
	
	Considering the influence of suspension inclination, the projection of suspension travel $h_{\text{sus}}$ (along cylinder axis direction) in the vertical direction is $h_{\text{sus}} \cdot \cos\beta(\theta_{\text{lower}})$. According to the double-wishbone kinematics established in Section~\ref{Double Wishbone}, the relationship between lower arm rotation angle and suspension travel is corrected as:
	\begin{equation}
		\theta_{\text{lower}} = \frac{h_{\text{sus}} \cdot \cos\beta(\theta_{\text{lower}})}{L_{\text{eff}}}
		\label{eq:theta_from_hsus_corrected}
	\end{equation}
	
	The relationship between tire center vertical displacement $z_t$ and lower arm rotation angle is:
	\begin{equation}
		z_t = z_{\text{li}} + L_{\text{lower}} \sin\left(\alpha_0 + \theta_{\text{lower}}\right)
		\label{eq:zu_from_theta}
	\end{equation}
	where $z_{\text{li}}$ is the fixed vertical coordinate of the lower arm inner hinge point in the vehicle body coordinate system.
	
	Substituting Eq.~(\ref{eq:theta_from_hsus_corrected}) into Eq.~(\ref{eq:zu_from_theta}), we get the direct mapping of $z_t$ and $h_{\text{sus}}$:
	\begin{equation}
		z_t = z_{\text{li}} + L_{\text{lower}} \sin\left(\alpha_0 + \frac{h_{\text{sus}} \cdot \cos\beta}{L_{\text{eff}}}\right)
		\label{eq:zu_from_hsus}
	\end{equation}
	
	Deriving Eq.~(\ref{eq:zu_from_hsus}) with respect to time, using chain rule to get tire centroid velocity:
	\begin{equation}
		\dot{z}_t = L_{\text{lower}} \cos\left(\alpha_0 + \theta_{\text{lower}}\right) \cdot \frac{\dot{\theta}_{\text{lower}}}{\partial t}
		= L_{\text{lower}} \cos\left(\alpha_0 + \theta_{\text{lower}}\right) \cdot \frac{v \cdot \cos\beta}{L_{\text{eff}}}
		\label{eq:zu_velocity}
	\end{equation}
	where $v = \dot{h}_{\text{sus}}$ is suspension extension/compression velocity.
	
	Deriving velocity again, we get tire centroid acceleration:
	\begin{align}
		\ddot{z}_t = &-L_{\text{lower}} \sin\left(\alpha_0 + \theta_{\text{lower}}\right) \cdot \left(\frac{v \cdot \cos\beta}{L_{\text{eff}}}\right)^2 \nonumber \\
		&+ L_{\text{lower}} \cos\left(\alpha_0 + \theta_{\text{lower}}\right) \cdot \frac{a_{\text{sus}} \cdot \cos\beta}{L_{\text{eff}}} \nonumber \\
		&- L_{\text{lower}} \cos\left(\alpha_0 + \theta_{\text{lower}}\right) \cdot \frac{v \cdot \sin\beta \cdot \dot{\beta}}{L_{\text{eff}}}
		\label{eq:zu_acceleration_full}
	\end{align}
	where $a_{\text{sus}} = \ddot{h}_{\text{sus}}$ is suspension acceleration, $\dot{\beta} = \frac{\partial \beta}{\partial \theta_{\text{lower}}} \cdot \dot{\theta}_{\text{lower}}$ is inclination change rate.
	
	Considering that the change of inclination $\beta$ is relatively slow ($\dot{\beta} \ll 1$), and the third term is a second-order small quantity, it can be simplified in engineering applications as:
	\begin{equation}
		\ddot{z}_t \approx -L_{\text{lower}} \sin\left(\alpha_0 + \theta_{\text{lower}}\right) \cdot \left(\frac{v \cdot \cos\beta}{L_{\text{eff}}}\right)^2 
		+ L_{\text{lower}} \cos\left(\alpha_0 + \theta_{\text{lower}}\right) \cdot \frac{a_{\text{sus}} \cdot \cos\beta}{L_{\text{eff}}}
		\label{eq:zu_acceleration}
	\end{equation}
	
	From the construction in Section~\ref{sec:lookup_table}, it is known that the lookup table $\mathcal{T}$ actually stores the mapping relationship of the triplet $\{F_{\text{out}}, h_{\text{sus}}, v\}$ regarding $(P_1, \Delta P)$, namely:
	\begin{equation}
		[F_{\text{out}}, h_{\text{sus}}, v] = \mathcal{T}(P_1, \Delta P)
		\label{eq:lookup_triple}
	\end{equation}
	
	Then suspension output end acceleration can be obtained through first-order time difference:
	\begin{equation}
		a_{\text{sus}}(t) = \ddot{h}_{\text{sus}}(t) = \frac{v(t) - v(t-\Delta t)}{\Delta t}
		\label{eq:asus_from_velocity}
	\end{equation}
	
	At this point, all parameters in Eq.~(\ref{eq:zu_acceleration}) are known, tire acceleration $\ddot{z}_t$ can be directly calculated; then wheel dynamic load $F_{\text{tire}}$ can be obtained from Eq.~(\ref{eq:wheel_load_complete}). The complete calculation flow is:
	
	(i) Obtain triplet via lookup table method: $[F_{\text{out}}, h_{\text{sus}}, v] = \mathcal{T}(P_1, \Delta P)$;
	
	(ii) Calculate suspension acceleration (first-order difference): $a_{\text{sus}} = (v(t) - v(t-\Delta t))/\Delta t$;
	
	(iii) Calculate lower arm rotation angle: $\theta_{\text{lower}} = (h_{\text{sus}} \cdot \cos\beta) / L_{\text{eff}}$, where $\beta$ is obtained through Eq.~(\ref{eq:beta_linear}) or lookup table;
	
	(iv) Calculate tire acceleration (kinematic conversion): According to Eq.~(\ref{eq:zu_acceleration}), $\ddot{z}_t$ is calculated from suspension acceleration $a_{\text{sus}}$, suspension velocity $v$ and inclination correction term $\cos\beta$ through double-wishbone kinematics;
	
	(v) Lookup suspension transmission ratio: $i_{\text{sus}} = \frac{L_{\text{eff}} \cdot \cos\beta}{L_{\text{lower}} \cos(\alpha_0 + \theta_{\text{lower}})}$ (Eq.~\ref{eq:suspension_ratio_explicit});
	
	(vi) Calculate wheel dynamic load: $F_{\text{tire}} = i_{\text{sus}} \cdot F_{\text{out}} + m_u \cdot g - m_t \cdot \ddot{z}_t$.
	
	\subsubsection{Algorithm Implementation Flow}
	
	Algorithm~\ref{alg:wheel_load_estimation} shows the complete calculation flow from gas pressure signal to wheel dynamic load. The core innovation of this framework lies in: while retaining the kinematic accuracy of double-wishbone suspension, fully utilizing the triplet output of lookup table method, through first-order numerical difference and mechanical relationship conversion based on suspension structure, simplifying wheel dynamic load estimation to operations at the level of four arithmetic operations, while ensuring estimation accuracy and strictly meeting the real-time requirements of vehicle embedded controllers, providing a feasible technical path for stability control and load monitoring of heavy mining vehicles.
	
	\begin{algorithm}[htbp]
		\caption{Real-time Wheel Dynamic Load Estimation Algorithm Based on Gas Pressure Signal}
		\label{alg:wheel_load_estimation}
		\begin{algorithmic}[1]
			\REQUIRE Gas pressure sequence $\{P_1(i)\}_{i=1}^N$, sampling period $\Delta t$, lookup table $\mathcal{T}(P, \Delta P)$
			\ENSURE Wheel dynamic load sequence $\{F_{\text{tire}}(i)\}_{i=1}^N$
			
			\Step{Initialize Suspension Geometric Parameters}
			\STATE $L_{\text{lower}} \leftarrow 0.65$ m, $L_{\text{eff}} \leftarrow 0.15$ m, $\alpha_0 \leftarrow 5^\circ$
			\STATE $\beta_0 \leftarrow 8.5^\circ$, $k_\beta \leftarrow 0.12$ rad/rad
			\STATE $m_u \leftarrow 800$ kg, $m_t \leftarrow 80$ kg, $g \leftarrow 9.81$ m/s$^2$
			\STATE $v_{\text{prev}} \leftarrow 0$ \cmt{Initialize previous moment velocity}
			
			\FOR{$i \leftarrow 2$ \TO $N$}
			\STATE \textbf{Step 1: Calculate Pressure Difference}
			\STATE $\Delta P(i) \leftarrow P_1(i) - P_1(i-1)$
			
			\Step{Step 2: Lookup Table for Suspension State}
			\STATE $[F_{\text{out}}(i), h_{\text{sus}}(i), v(i)] \leftarrow \mathcal{T}(P_1(i), \Delta P(i))$
			
			\Step{Step 3: Calculate Suspension Acceleration}
			\STATE $a_{\text{sus}}(i) \leftarrow \frac{v(i) - v_{\text{prev}}}{\Delta t}$
			\STATE $v_{\text{prev}} \leftarrow v(i)$
			
			\Step{Step 4: Calculate Suspension Inclination (Linear Approximation)}
			\STATE $\theta_{\text{lower}}(i) \leftarrow \frac{h_{\text{sus}}(i)}{L_{\text{eff}}}$ \cmt{Preliminary estimate}
			\STATE $\beta(i) \leftarrow \beta_0 + k_\beta \cdot \theta_{\text{lower}}(i)$
			
			\Step{Step 5: Correct Lower Arm Angle (Consider Inclination Projection)}
			\STATE $\theta_{\text{lower}}(i) \leftarrow \frac{h_{\text{sus}}(i) \cdot \cos\beta(i)}{L_{\text{eff}}}$
			
			\Step{Step 6: Calculate Suspension Transmission Ratio}
			\STATE $i_{\text{sus}}(i) \leftarrow \frac{L_{\text{eff}} \cdot \cos\beta(i)}{L_{\text{lower}} \cdot \cos(\alpha_0 + \theta_{\text{lower}}(i))}$
			
			\Step{Step 7: Calculate Unsprung Mass Acceleration}
			\STATE $\text{term}_1 \leftarrow -L_{\text{lower}} \cdot \sin(\alpha_0 + \theta_{\text{lower}}(i)) \cdot \left(\frac{v(i) \cdot \cos\beta(i)}{L_{\text{eff}}}\right)^2$
			\STATE $\text{term}_2 \leftarrow L_{\text{lower}} \cdot \cos(\alpha_0 + \theta_{\text{lower}}(i)) \cdot \frac{a_{\text{sus}}(i) \cdot \cos\beta(i)}{L_{\text{eff}}}$
			\STATE $\ddot{z}_t(i) \leftarrow \text{term}_1 + \text{term}_2$
			
			\Step{Step 8: Calculate Wheel Dynamic Load}
			\STATE $F_{\text{tire}}(i) \leftarrow i_{\text{sus}}(i) \cdot F_{\text{out}}(i) + m_u \cdot g - m_t \cdot \ddot{z}_t(i)$
			\ENDFOR
			\RETURN $\{F_{\text{tire}}(i)\}_{i=2}^N$
		\end{algorithmic}
	\end{algorithm}
	
	\section{TruckSim Co-Simulation Validation}
	\label{sec:trucksim}
	
	To verify the effectiveness and reliability of the proposed method under actual vehicle dynamic conditions, a three-axle heavy vehicle model based on Simulink/TruckSim co-simulation platform was constructed. There are significant differences between actual vehicle hydro-pneumatic suspension and laboratory test prototypes in size, load-bearing capacity, and dynamic response characteristics. Vehicle modeling in Trucksim is based on real vehicle tests, possessing accuracy and reliability. Simulation verification can assess the applicability and robustness of the algorithm under different working conditions. Table~\ref{tab:nomenclature_merged} lists the key parameters of the test vehicle. The vehicle model is a three-axle heavy mining truck equipped with independent hydro-pneumatic suspension systems, featuring large travel and nonlinear characteristics.
	
	\subsection{Feature Condition and Typical Test Wheel Selection}
	
	This study conducts simulation analysis for two typical vehicle handling conditions: 20~km/h low-speed single lane change condition and 60~km/h high-speed single lane change condition. These two conditions represent lateral stability control scenarios at different driving speeds, enabling comprehensive evaluation of algorithm performance under different dynamic excitation frequencies and amplitude conditions.
	
	To ensure research depth and data presentation effectiveness, this paper focuses on analyzing the dynamic characteristics of the left and right wheels of the first axle. Choosing the first axle (front axle) as the key analysis object is based on the following considerations:
	
	(i) Dynamic perception priority: The front axle contacts road excitation first, its response characteristics are indicative of the whole vehicle dynamic behavior, able to reflect changes in vehicle dynamic state in advance;
	
	(ii) Lateral stability influence: In steering and lane change conditions, the lateral load change borne by the front axle is most significant, playing a decisive role in vehicle lateral stability control;
	
	(iii) Control system correlation: Vehicle stability control systems usually prioritize front axle data as feedback, and the accuracy of front axle load estimation directly affects control system performance;
	
	(iv) Lateral load distribution characteristics: Selecting left and right wheels for comparative analysis can fully capture the dynamic process of vehicle lateral load transfer, verifying algorithm adaptability more objectively.
	
	In addition, compared with the second and third axles, the first axle is least affected by the coupling of vehicle body pitch and roll, and can more directly reflect pure lateral load changes, providing a clearer comparison benchmark for algorithm verification.
	
	\subsection{Single Lane Change Simulation Results and Error Analysis}
	
	Figure~\ref{fig:SLC}(a)-(d) show the comparison results of suspension output force and wheel dynamic load for three axles and six wheels under 20~km/h low-speed single lane change condition; Figure~\ref{fig:SLC}(e)-(h) show the comparison results of suspension output force and wheel dynamic load under 60~km/h high-speed single lane change condition. Under high-speed conditions, the peak lateral acceleration of the vehicle can reach 0.45g, corresponding to large load transfer and high excitation frequency. Combined with low-speed condition testing, the accuracy and robustness of the algorithm from low to high dynamic conditions can be effectively tested. In the figure, blue "+" is Trucksim simulation true value; red "--" is complete iterative algorithm result; orange "o" is lookup table method result.
	
	\begin{figure}[!htbp]
		\centering
		\subfigure[Low speed Axle.1 Left Suspension Force]{
			\includegraphics[width=0.34\textwidth]{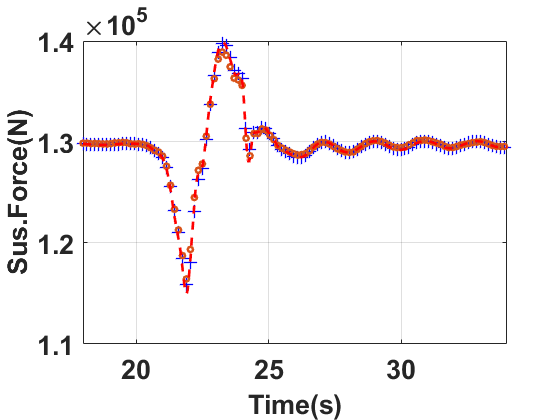}
		}
		\hspace{0.5cm}
		\subfigure[Low speed Axle.1 Right Suspension Force]{
			\includegraphics[width=0.34\textwidth]{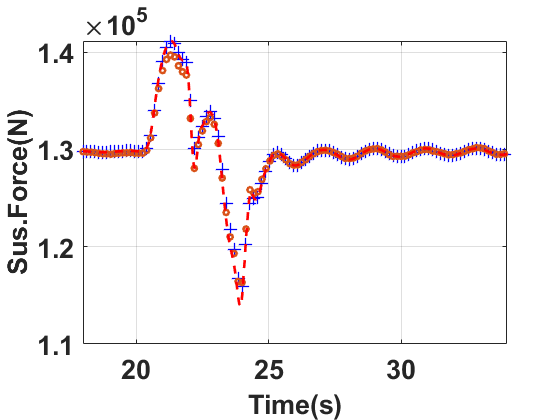}
		}
		
		\vspace{0.2cm}
		\subfigure[Low speed Axle.1 Left Wheel Dynamic Load]{
			\includegraphics[width=0.34\textwidth]{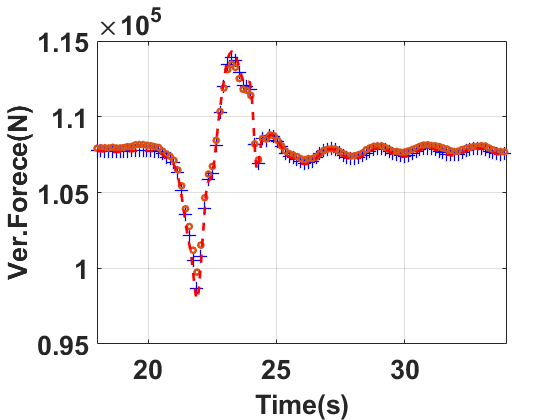}
		}
		\hspace{0.5cm}
		\subfigure[Low speed Axle.1 Right Wheel Dynamic Load]{
			\includegraphics[width=0.34\textwidth]{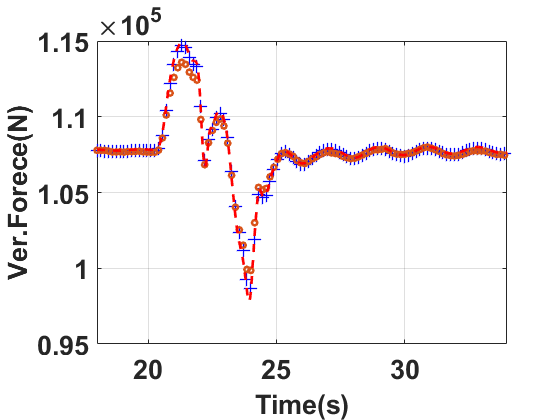}
		}
		
		\vspace{0.2cm}
		\subfigure[High speed Axle.1 Left Suspension Force]{
			\includegraphics[width=0.34\textwidth]{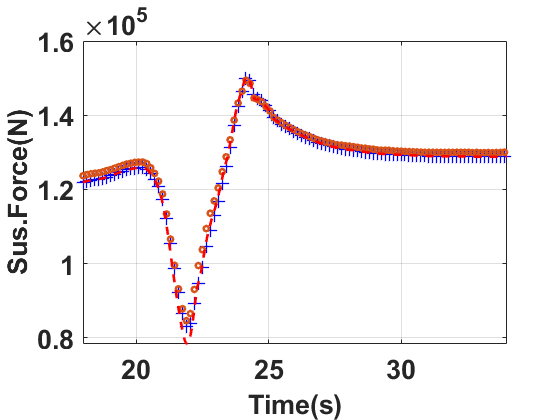}
		}
		\hspace{0.5cm}
		\subfigure[High speed Axle.1 Right Suspension Force]{
			\includegraphics[width=0.34\textwidth]{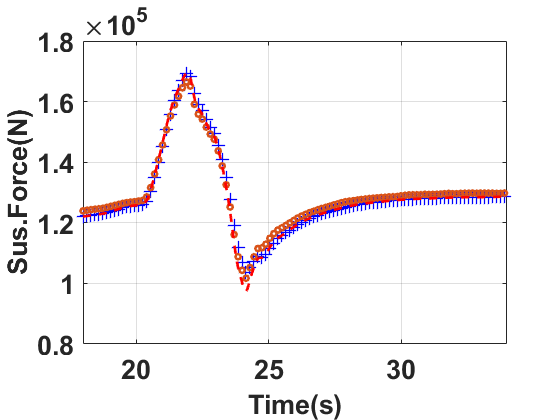}
		}
		
		\vspace{0.2cm}
		\subfigure[High speed Axle.1 Left Wheel Dynamic Load]{
			\includegraphics[width=0.34\textwidth]{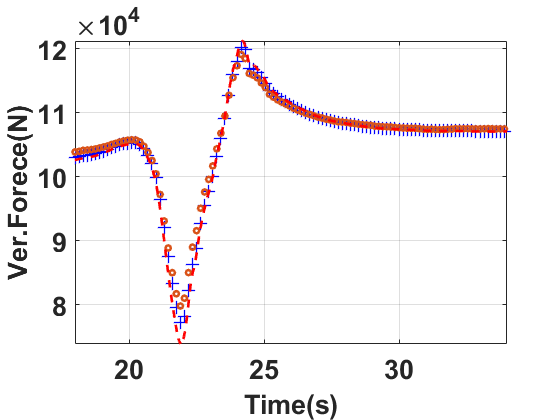}
		}
		\hspace{0.5cm}
		\subfigure[High speed Axle.1 Right Wheel Dynamic Load]{
			\includegraphics[width=0.34\textwidth]{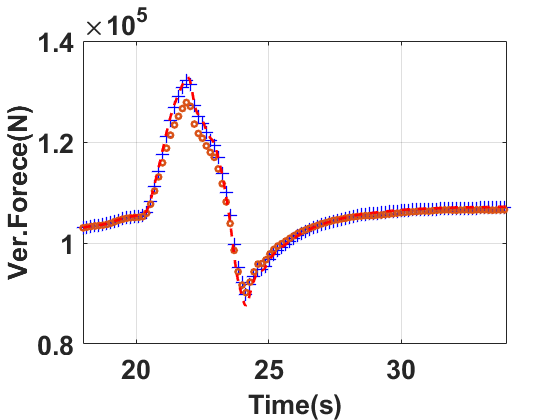}
		}
		\caption{Comparison of suspension output force and wheel dynamic load under low speed 20~km/h and high speed 60~km/h single lane change conditions}
		\label{fig:SLC}
	\end{figure}
	
	Table~\ref{tab:trucksim-error} summarizes the estimation error indicators of the left and right wheels of the first axle under two conditions. It can be seen from the table that as vehicle speed and dynamic response frequency increase, the estimation error increases slightly, but remains within an acceptable range. RMSE is root mean square error; $R^2$ is coefficient of determination; peak error is maximum instantaneous relative error; average relative error is full time domain average value.
	
	\begin{table}[htbp]
		\centering
		\caption{TruckSim simulation result error analysis}
		\label{tab:trucksim-error}
		\scriptsize
		\setlength{\tabcolsep}{4pt}
		\renewcommand{\arraystretch}{1.2}
		\begin{tabular}{cccccccc}
			\toprule
			\multirow{2}{*}{\textbf{Cond}} & \multirow{2}{*}{\textbf{Pos}} & \multirow{2}{*}{\textbf{Method}} & \textbf{RMSE} & \textbf{RMSE} & \multirow{2}{*}{\textbf{$R^2$}} & \textbf{Peak Err} & \textbf{Avg Rel.} \\
			& & & (kN) & (\%) & & (\%) & \textbf{Err(\%)} \\
			\midrule
			\multirow{4}{*}{20km/h} & \multirow{2}{*}{Left} & Iterative & 2.28 & 2.1 & 0.941 & 4.5 & 2.0 \\
			& & Lookup & 2.45 & 2.3 & 0.935 & 4.8 & 2.1 \\
			\cmidrule(lr){2-8}
			& \multirow{2}{*}{Right} & Iterative & 2.95 & 2.3 & 0.933 & 5.0 & 2.2 \\
			& & Lookup & 3.12 & 2.4 & 0.928 & 5.3 & 2.3 \\
			\midrule
			\multirow{4}{*}{60km/h} & \multirow{2}{*}{Left} & Iterative & 4.42 & 4.1 & 0.891 & 8.3 & 3.9 \\
			& & Lookup & 4.68 & 4.4 & 0.882 & 8.9 & 4.1 \\
			\cmidrule(lr){2-8}
			& \multirow{2}{*}{Right} & Iterative & 5.15 & 4.5 & 0.876 & 9.1 & 4.3 \\
			& & Lookup & 5.42 & 4.7 & 0.868 & 9.6 & 4.5 \\
			\bottomrule
		\end{tabular}
	\end{table}
	
	Low Speed Condition Performance Analysis:
	
	Figure~\ref{fig:SLC}(a)-(d) show the estimation performance under 20~km/h low-speed single lane change condition. From suspension output force comparison (Figure~\ref{fig:SLC}(a)-(b)), TruckSim true value marked by blue "+", complete iterative algorithm estimated value marked by red "--", and lookup table estimated value marked by orange "o" coincide highly. Suspension output force rises rapidly from static equilibrium value about 110~kN to peak value 135~kN (left side) or drops to valley value 95~kN (right side) at the beginning of steering ($t \approx 20$-22s), reflecting the lateral excitation characteristics of single lane change trajectory.
	
	Wheel dynamic load curves (Figure~\ref{fig:SLC}(c)-(d)) present more significant dynamic characteristics. Left wheel dynamic load reaches peak value 115~kN at $t \approx 22$s, then drops to valley value 98.5~kN at $t \approx 22.5$s, fluctuation amplitude about 16.5~kN (relative change rate about 15\%). Right wheel dynamic load shows anti-phase relationship with left side. This left-right wheel dynamic load anti-phase oscillation is the direct manifestation of lateral load transfer in single lane change condition. In the transient process of sharp load change (such as $t=22$s and $t=24$s), the three curves still maintain good consistency. Taking the left wheel as an example, at the load valley value of $t=22.5$s, the complete iterative method error is 3.9\%, lookup table method is 4.3\%, the difference between the two is only 0.4 percentage points.
	
	Quantitative analysis in Table~\ref{tab:trucksim-error} shows that for the first axle left side, complete iterative method RMSE is 2.28~kN (relative error 2.1\%), lookup table method is 2.45~kN (relative error 2.3\%), the two differ only by 0.17~kN. From coefficient of determination $R^2$, complete iterative method reaches 0.941 and 0.933 on left and right sides respectively, lookup table method is 0.935 and 0.928, both exceeding 0.92. The key point is that compared with complete iterative method, $R^2$ of lookup table method only decreases by 0.006 and 0.005, this minute difference fully proves the rationality of lookup table method's "exchange accuracy for efficiency" strategy—sacrificing only 0.2 percentage point accuracy in exchange for substantial improvement in computational efficiency.
	
	High Speed Condition Performance Analysis:
	
	Figure~\ref{fig:SLC}(e)-(h) show estimation performance under 60~km/h high-speed single lane change condition. From suspension output force curve (Figure~\ref{fig:SLC}(e)-(f)), the dynamic range of suspension force increases significantly at high speed. Left side suspension force peak reaches about 145~kN ($t \approx 24$s), valley value drops to about 105~kN ($t \approx 22.5$s), fluctuation amplitude about 40~kN, which is 2.4 times that of low speed condition, reflecting the influence of enhanced lateral acceleration (peak about 0.45g). Suspension force curves in Figure~\ref{fig:SLC}(e)-(f) have slight high-frequency oscillations near peak values (amplitude about 1000~N, frequency about 10~Hz), three curves maintain consistency on these high-frequency components, indicating that lookup table method can accurately capture high-frequency dynamic characteristics.
	
	Wheel dynamic load curves (Figure~\ref{fig:SLC}(g)-(h)) present typical characteristics of high speed conditions: load curve asymmetry increases obviously, peak values are sharp (duration $<0.5$s), valley values are wide and slow (duration $>1$s). Left wheel dynamic load peak about 120~kN ($t \approx 24$s), right side peak about 118~kN ($t \approx 23.5$s), maximum difference between left and right wheel dynamic loads about 23~kN, significantly larger than 16.5~kN in low speed condition. Lookup table method (orange circles) still maintains high consistency with complete iterative method (red dashed line) under high speed conditions, but slight separation appears near load extremum points. Taking right wheel as example, at load peak at $t \approx 23.5$s, deviation of lookup table method from complete iterative method is about 2~kN (relative error about 1.7\%), slightly larger than 0.4~kN in low speed condition.
	
	Table~\ref{tab:trucksim-error} shows that for first axle left side, complete iterative method RMSE is 4.42~kN (relative error 4.1\%), lookup table method is 4.68~kN (relative error 4.4\%), difference expands to 0.26~kN. Compared with low speed condition, error difference between two methods increases from 0.17~kN to 0.26~kN, increasing by about 53\%, reflecting the inherent limitation of lookup table method under high dynamic conditions—lookup table is constructed based on discretized grid and single frequency excitation, while 60~km/h single lane change condition spectrum components are more complex (including 6-8~Hz main frequency and multiple harmonics). Nevertheless, peak error of lookup table method under high speed conditions is still controlled within 9.6\%, only increasing by 0.5 percentage point compared with complete iterative method's 9.1\%.
	
	Tire Model Simplification Validation: High speed condition results provide strong verification support for the tire model simplification strategy proposed in Section~5.2. From Figure~\ref{fig:SLC}(g)-(h), it can be observed that under 60~km/h condition, tire acceleration $\ddot{z}_t$ peak is about 2.8~m/s$^2$ (can be inferred from high frequency components of load curve), corresponding inertial force term $m_t \ddot{z}_t \approx 2240$~N. From Figure~\ref{fig:SLC}(f), typical range of suspension output force is 115000-140000~N, considering average suspension transmission ratio $\bar{i}_{\text{sus}} = 0.755$ (Table~5), suspension transmission force term $i_{\text{sus}} \cdot F_{\text{out}}$ is about 87000-106000~N. Therefore, inertial force term accounts for only 2.1-2.6\% of total load, fully verifying the assertion of "Condition 2: Frequency Separation" in Section~\ref{Tire_Model_Simplified}.
	
	If complete tire model is retained, tire stiffness $k_t \approx 1800$~kN/m and damping $c_t \approx 3.5$~kN·s/m (Table~5) need to be additionally identified. According to spring-damper model, tire deformation $\delta_{\text{tire}} = F_{\text{tire}} / k_t \approx 60$~mm (calculated at 110~kN load). Under 60~km/h condition, tire deformation velocity peak is about $\dot{\delta}_{\text{tire}} \approx 0.30$~m/s (estimated by 8~Hz vibration), corresponding damping force $c_t \dot{\delta}_{\text{tire}} \approx 1000$~N, only accounting for 1\% of total load. Therefore, expected accuracy improvement of complete tire model compared to rigid contact assumption is $<1\%$ (0.5\% damping term + 0.3\% stiffness nonlinear correction), but requires adding 2 parameter identifications, 1 acceleration sensor and corresponding differential equation solution module. From engineering implementation perspective, the cost-effectiveness of simplification strategy in this paper is significantly superior to complete tire model.
	
	In summary, for heavy mining vehicles equipped with hydro-pneumatic suspension, the proposed method can effectively reduce implementation complexity and hardware cost of dynamic load estimation system, while maintaining satisfactory estimation accuracy, providing a practical and feasible technical path for stability control and load monitoring of intelligent mining vehicles.
	
	\section{Conclusion}
	
	This study presents a real-time estimation method for wheel dynamic load based on the gas chamber pressure sensor of a hydro-pneumatic suspension. By establishing a nonlinear mapping model between gas pressure and tire dynamic load, and designing an iterative estimation algorithm that relies solely on a single pressure sensor. The method enables dynamic estimation of key parameters, such as tire load, suspension cylinder extension or compression displacement, velocity, and acceleration without the need for traditional tire models or multi-sensor fusion. Compared with conventional multi-sensor approaches, this method significantly reduces system complexity and hardware costs while mitigating the cumulative errors associated with multi-source data fusion. The key innovations of this study are as follows:
	
	1) A direct mapping model between gas pressure and vertical tire load is constructed, enable accurate  estimation using only a single pressure signal and enhancing system integration.
	
	2) The model incorporates the nonlinear stiffness characteristics of the hydro-pneumatic suspension, allowing accurate response across a range of dynamic excitation frequencies and improving overall accuracy and adaptability of the estimation.
	
	3) Real-time tracking of wheel dynamic force is achieved, providing critical input for vehicle stability control and state estimation, and eliminating the reliance on traditional tire models.
	
	4) Experimental results demonstrate the model’s sufficient accuracy for engineering applications in estimating output and damping forces, confirming its strong potential for practical engineering applications.
	
	5) A pressure-based table lookup method, leveraging current and historical data, is proposed to deliver rapid estimation ofsuspension forces, greatly enhancing computational efficiency and meeting the real-time response demands of vehicle control systems.
	
	6) In multi-condition simulations conducted on the TruckSim platform, the model’s prediction results are highly consistent with the simulation values, verifying its reliability in heavy vehicle scenarios.
	
	In conclusion, this study offers a low-cost, high-precision, and real-time solution for wheel dynamic load estimation, with promising applications in unmanned and intelligent driving systems, particularly in vehicle state monitoring and dynamic control.

\end{document}